\documentclass[letterpaper, 10 pt, conference]{ieeeconf}  % Comment this line out if you need a4paper

\IEEEoverridecommandlockouts                              % This command is only needed if 
\usepackage[utf8]{inputenc}
\usepackage[T1]{fontenc}
\usepackage[hidelinks]{hyperref}
\usepackage{url}
\usepackage{ifthen}
\usepackage{cite}
\usepackage[cmex10]{amsmath} % Use the [cmex10] option to ensure complicance
\usepackage{amsfonts,amssymb,amsmath,bbm}
\usepackage{graphicx}
\usepackage{color}
\usepackage{ifthen}
\usepackage{comment}

\usepackage{tikz,pgfplots}
\pgfplotsset{width=7cm,compat=1.3}

\usepackage[justification=centering]{caption}

\usepackage[ruled,vlined]{algorithm2e}

\usepackage[arxiv]{optional}  % arxiv/submission

\usepackage{enumerate}

\allowdisplaybreaks[1]

\newtheorem{lemma}{Lemma}

\newtheorem{theorem}{Theorem}
\newtheorem{definition}{Definition}

\usepackage{caption}
\usepackage{subcaption}

\newcommand{\map}[3]{#1: #2 \rightarrow #3}
\newcommand{\signal}{\pi}
\newcommand{\RR}{\mathbb{R}}
\newcommand{\supp}{\mathrm{supp}}
\newcommand{\de}{\mathrm{d}}

\usepackage{mathtools}
\title{Information Design for a Non-atomic Service Scheduling Game}

\author{Nasimeh Heydaribeni, Ketan Savla\thanks{Nasimeh Heydaribeni is with the University of Michigan, Ann Arbor, MI. \texttt{heydari@umich.edu}.}\thanks{Ketan Savla is with the University of Southern California, Los Angeles, CA. \texttt{ksavla@usc.edu}.}}

\date{\today}

\begin{document}

	\maketitle
	\thispagestyle{empty}
	\pagestyle{empty}
	
	%\author{\IEEEauthorblockN{Nasimeh Heydaribeni}
	%		\IEEEauthorblockA{\textit{Department of Electrical and Computer Engineering} \\
	%			\textit{University of Michigan, Ann Arbor, MI, USA} \\
	%			heydari@umich.edu}
	%		\and
	%		\IEEEauthorblockN{Achilleas Anastasopoulos}
	%		\IEEEauthorblockA{\textit{Department of Electrical and Computer Engineering} \\
	%			\textit{University of Michigan, Ann Arbor, MI, USA} \\
	%			anastas@umich.edu}

	\maketitle
	
	\begin{abstract}
		We study an information design problem for a non-atomic service scheduling game. The service starts at a random time and there is a continuum of agent population who have a prior belief about the service start time but do not observe the actual realization of it. The agents want to make decisions of when to join the queue in order to avoid long waits in the queue or not to arrive earlier than the service has started. There is a planner who knows when the service starts and makes suggestions to the agents about when to join the queue through an obedient direct signaling strategy, in order to minimize the average social cost. We characterize the full information and the no information equilibria and we show in what conditions it is optimal for the planner to reveal the full information to the agents. Further, by imposing appropriate assumptions on the model, we formulate the information design problem as a generalized problem of moments (GPM) and use computational tools developed for such problems to solve the problem numerically.
		
	\end{abstract}
	
	%	\begin{keywords}
	%linear quadratic Gaussian (LQG) games, perfect Bayesian  equilibrium (PBE), dynamic games, asymmetric information.
	%	\end{keywords}

	%%%%%%%%%%%%%%%%%%%%%%%%%%%%%%%%%%%%%%%%%%%%%%%%%%%%%%%%%%%%%%%%%%%%%%%
	%%%%%%%%%%%%%%%%%%%%%%%%%%%%%%%%%%%%%%%%%%%%%%%%%%%%%%%%%%%%%%%%%%%%%%%
	%%%%%%%%%%%%%%%%%%%%%%%%%%%%%%%%%%%%%%%%%%%%%%%%%%%%%%%%%%%%%%%%%%%%%%%
	%%%%%%%%%%%%%%%%%%%%%%%%%%%%%%%%%%%%%%%%%%%%%%%%%%%%%%%%%%%%%%%%%%%%%%%

	\section{Introduction}
	Information asymmetry is inevitable in today's ever-growing systems and networks. Each agent in these systems faces decision makings in the presence of uncertainty toward some states of the world or other agent's information \cite{HeAn19,heydaribeni2020structured}. Having access to as much information as possible enables these agents to make more profitable decisions.  Information design  \cite{kamenica2011bayesian,bergemann2019information} studies how sharing information strategically with the agents can steer their actions towards a desirable direction. 
	In the information design  framework, there is a sender that possesses some private knowledge about the state of the world. There are possibly multiple receivers of the information. The information that the sender shares with the receivers  is shaped in a way to align their objectives with that of the sender as much as possible. 
	%Note that the receivers will act strategically upon receiving the information, but the sender's objective is to affect the receivers' best responses to favor in his direction. 

	There are different types of information design problems depending on whether there are multiple receivers or a single one, whether the system is dynamic or not, whether the receivers have private information or not, etc. The information design problems with a single receiver are  referred to as ``Bayesian persuasion'' as first introduced in \cite{kamenica2011bayesian}, where a geometric method of analyzing the information design problem is proposed. 
	Information design problems with more than one receiver are usually more complex since the solution must induce an   equilibrium between the receivers. It is shown in \cite{bergemann2016bayes} that the set of outcomes in an information design problem with multiple receivers is indeed the set of Bayes-correlated equilibria, BCE, for the receivers.  According to the definition of BCE in \cite{bergemann2016bayes}, the information shared with the receivers contains suggestions of what actions they should take. Therefore, an obedience condition has to be imposed on the strategy of the sender to make sure the receivers will follow the suggestions once they hear them. The obedience condition is the same as the conditions that are imposed in the definition of correlated equilibria. 
	
	Information design problems study dynamic or static systems. In static information design, the problem that the information designer faces is a static optimization problem \cite{kamenica2011bayesian,bergemann2016bayes,bergemann2016information,das2017reducing,tavafoghi2017informational,zhu2018stability}. Dynamic information design problems\cite{ely2017beeps,farhadi2020dynamic,doval2020sequential,lingenbrink2019optimal, sayin2018dynamic,ARXIV-VERSION_HeAn21} deal with dynamic settings and therefore, the  information can be disclosed sequentially. Dynamic programming techniques can therefore  be used to characterize the optimal strategy.

	An example of dynamic information design can be found in  \cite{ely2017beeps}, where the receiver is awaiting the occurrence of a random event, e.g., the arrival of an email, so that she can check her email.  The receiver is informed of the arrival of the email by a beep that is sent by the sender. The sender wants the receiver not to check her email for as long as possible. Therefore,  the sender has to solve a dynamic information design problem to decide whether or not or when to send a beep and reveal the arrival of an email. The problem is solved in continuous time and then a discrete time generalization is presented.

	In this paper, we study an information design problem where there are not only multiple receivers, but they are non-atomic. That is, they form a continuum of population with unit total mass. A service scheduling problem is studied where the service start time is unknown to the agents who want to make decisions of when to join the queue in order to avoid long waits in the queue or not to arrive earlier than the service has started. The service starting time and agents' decisions are in continuous time. There is a planner that knows when the service starts and makes suggestions to the agents about when to join the queue. The suggestion profile has to satisfy the obedience condition. That is, an agent that has received the suggestion of joining at time $t$ must be willing to obey that suggestion. Our model can be considered a dynamic information design problem because the planner makes suggestions for the whole dynamic arrival process of agents. However, each agent only receives one signal from the planner.

	Our model of a continuum of agent population arriving at a queue and their cost function closely follows that of \cite{smith1984existence,daganzo1985uniqueness}. The existence and uniqueness of the equilibrium arrival process is proved in \cite{smith1984existence} and \cite{daganzo1985uniqueness}, respectively.  In these works, the agents have a preference of when to depart the queue while in our model, this preferred time coincides with the time the service starts and is also the same for all of the agents, although they do not know when that time is. This is where the information design aspect of our model plays its role.
	Information design for non-atomic agents has also been studied in \cite{zhu2020semidefinite}, where a routing game has been 
	considered in which the unknown states of the world affect the latency of the links. The problem has been shown to be a generalized problem of moments and a  hierarchy of polynomial optimization is proposed to approximate the solution.

	The contributions of this paper are as follows. We formulate an information design problem for a service scheduling game consisting of non-atomic agents. We characterize the equilibrium in full information and no information extremes. We show some results on when the planner can do no better than revealing the full information to the agents. We impose some assumptions on our model that will allow us to express the information design problem as a generalized problem of moments (GPM) \cite{lasserre2008semidefinite}. We use the computation tools for these problems  
	such as Gloptipoly \cite{henrion2009gloptipoly} to numerically solve the information design problem.

	The rest of the paper is structured as follows. Problem formulation is discussed  in section \ref{model}. In section \ref{infodes}, the obedience condition is defined and simplified.  We study the two extreme cases of full information and no information equilibria in section \ref{extreme}. A structural result is stated in section \ref{structural} for a specific type of  arrival processes. We formulate the problem as a GPM in section \ref{GPM} and we present numerical analysis in section \ref{numerical}. We conclude in section \ref{conclusion}.   \optv{submission}{The proofs of some of the lemmas and theorems can be found in the Appendix at the end of the paper. Due to space limitations, some of the  proofs can be found in the extended version of this paper \cite{ARXIV-VERSION_HeSa21}.}\optv{arxiv}{The proofs of the lemmas and theorems can be found in the Appendix at the end of the paper.}

	\section{Problem Formulation}\label{model}
	A service provider starts its service at a fixed rate $\mu \in (0,1)$ starting at some time $\tau \geq 0$ with probability distribution of $f_{\tau}(\cdot)$. 
	A continuum of agent population of unit total mass needs this service. 
	The action of an agent is the time $t$ to join the service queue. The collection of actions of all the agents can be represented as a probability measure, $m$, on $\RR_{\geq 0}$. Let the set of such measures be denoted as $M$. We usually refer to the measure $m$ as the arrival process.
	Note that the support of $m$ can include negative numbers. That is, the arrival times of agents can be a negative number which is due to the fact that the time origin is considered to be when the service can possibly start. 
	For a given $m \in M$ and $\tau$, we denote the queue length at time $t$ by $q_{\tau,m}(t)$, which is given as follows.
	\begin{align*}
	q_{\tau,m}(t)=\int_{s=-\infty}^t m(s)ds-\mu (t-\tau)^+,
	\end{align*}
	where $(a)^+=\max(a,0)$.
	
	The cost of an agent with action $t \in \supp(m)$ and for a given $\tau$ and $m$ is denoted by $c_{\tau,m}(t)$  and is the weighted sum of (i) time to wait in the queue until receiving service; and (ii) the difference between the time of service and realization of $\tau$. Note that (i) includes the time to wait for the service to start in case $t<\tau$. Part (ii)  is considered to capture the possibility of service quality deterioration by time. For example, the service quality degrades by time in  a food distribution center since the food quality degrades. Part (ii) also captures the fact that agents might be impatient and want to get serviced as soon as possible. 
	Therefore, we have the following cost function.  
	\begin{align*}
	c_{\tau,m}(t)=&c_1(\frac{q_{\tau,m}(t)}{\mu}+(\tau-t)^+) \nonumber \\&+c_2(t+\frac{q_{\tau,m}(t)}{\mu}+(\tau-t)^+-\tau)\\
	\equiv& \frac{q_{\tau,m}(t)}{\mu}+c(\tau-t)^++(1-c)(t-\tau)^+\\
	=&  \frac{q_{\tau,m}(t)}{\mu}+(t-\tau)^+-c(t-\tau),
	\end{align*}
	where $c=\frac{c_1}{c_1+c_2}\leq 1$, and $c_1$ and $c_2$ are the weights of the two parts of the cost function. Without loss of generality, we can assume $c_1$ and $c_2$ are between 0 and 1.
	
	The social cost associated with an arrival process $m$ and a $\tau$ is denoted by $s(\tau,m)$ and  is defined as the sum of costs of all agents, i.e.,  $s(\tau,m):=\int_{t} m(t) c_{\tau,m}(t) \de t.$

	The service rate $\mu$ and the probability distribution of $\tau$, $f_{\tau}(\cdot)$, are common knowledge. The agents do not know the exact realization of $\tau$, but there is a  planner who does. The planner desires to utilize this information asymmetry to minimize expected social cost over all \emph{obedient} direct signaling strategies. A direct signaling strategy is a map $\map{\signal}{\RR_{\geq 0}}{\triangle\, M}$, where $\triangle \, M$ is the set of probability distributions over $M$. That is, for a realization $\tau$, the planner  privately recommends actions to the agents consistent with a measure $m \in M$ sampled from $\signal(.|\tau)$. 	The obedience condition is defined in the next section.
	The objective of the planner is to minimize the average value of the social costs, $\bar{s}(\pi)$, which  is given below.  
	\begin{subequations}
		\begin{align*}
		&\bar{s}(\pi):=\int_{\tau,m}\int_{t} m(t) c_{\tau,m}(t) f_{\tau}(\tau)\pi(m|\tau) \de m \de \tau \de t.
		%\\&=\int_{\tau,m}\int_{t} m(t) ( \frac{q_{\tau,m}(t)}{\mu}+(t-\tau)^+-c(t-\tau)) f_{\tau}(\tau)\nonumber \\& \hspace{4.5cm} \pi(m|\tau)\de m \de \tau \de t
		\end{align*}
	\end{subequations}
	Throughout the paper, we impose different assumptions on the set of arrival processes $M$, to which the designer restricts his attention. In each section, it will be stated which assumption has been considered. Below is the list of these assumptions. 
	
	\textbf{Assumptions:}
	\begin{enumerate}[(a)]
		\item $m(t)\leq \mu, \forall t$. \label{lessmu}
		\item For all $m$ with $\pi(m|\tau)>0$, $m(t)=0$ for $t \notin  [\underline{t}_{\tau},\bar{t}_{\tau}]$ and some $\underline{t}_{\tau}$ and $\bar{t}_{\tau}$ that are increasing with respect to $\tau$. \label{interval}
		\item For all $m$ with $\pi(m|\tau)>0$, if  $q_{\tau,m}(t)=0$ and $m(s)>0$ for some $s>t$ and $s<t$, then $m(t)=\mu$. \label{qpos}
		\item $m$ is piecewise continuous. \label{piece}
	\end{enumerate}
	Note that assumption (\ref{qpos}) is to make sure that the  server  works at its full capacity as long as there is yet agents to arrive. As we will see in section \ref{extreme}, the full information equilibrium arrival process satisfies all of the above assumptions. Further, the no information equilibrium arrival process that satisfies (\ref{piece}), also satisfies assumptions (\ref{lessmu}) and (\ref{interval}).
	%\end{assumption*}

	\section{Obedience Condition}\label{infodes}
	%	The  planner will communicate some information about $\tau$ through the signaling strategy $\pi(m|\tau)$. That is, after observing $\tau$, the suggestion profile (or arrival process, if obeyed) $m(\cdot)$ is generated and accordingly, arriving time suggestions $t$, are given to the agents. 
	
	%We assume that $m(t)\leq \mu$ for all $t$ that results in $\tilde{\tau}_m(t)$ being increasing wrt $t$. 
	The agent that has received suggestion $t$, will  form her posterior belief on $\tau$ and $m$ which can be used to calculate the average cost of taking action $s$ (arriving at time $s$). We denote this average cost by $\bar{c}_{t,\pi}(s)$. The posterior belief of an agent that has received the suggestion $t$ is given below. 
	\begin{align*}
	\beta(\tau,m|t,\pi)=\frac{ f_{\tau}(\tau)\pi(m|\tau)m(t)}{\int_{\tau,m} f_{\tau}(\tau)\pi(m|\tau)m(t)}.
	\end{align*}

	In order to calculate $\bar{c}_{t,\pi}(s)$, we define a quantity $\tilde{\tau}_m(t)$ as follows.
	For a given arrival process $m$  and each $t\geq 0$, we define $\tilde{\tau}_m(t)\leq t$ as follows.
	\begin{align*}
	\forall \tau\leq \tilde{\tau}_m(t), \quad q_{\tau,m}(t)=0\\
	\forall \tau > \tilde{\tau}_m(t), \quad q_{\tau,m}(t)>0.
	\end{align*}
	Note that there might exist a $t$ for which we have $q_{\tau,m}(t)>0$, for all $\tau$. In this case, we define  $\tilde{\tau}_m(t)=0$. Also, for $t<0$, we define $\tilde{\tau}_m(t)=0$. 
	
	Throughout the paper, except for section \ref{extreme}, we assume $m$ follows assumption (\ref{lessmu}).  As we will see in section \ref{extreme}, both full information and no information equilibria satisfy this assumption.  
	
	The average cost $\bar{c}_{t,\pi}(s)$  is given in the following lemma.
	\begin{lemma}\label{avcost}
		$\bar{c}_{t,\pi}(s)$ which is the average value of the cost for an agent that has received suggestion $t$ through the signaling strategy $\pi$ is given as follows.
		\begin{align*}
		&\bar{c}_{t,\pi}(s)=\mathbb{E}\{c_{\tau,m}(s)|t,\pi\}= 
		\\ &
		\frac{1}{\mu \bar{m}(t)}\int_{
			m, \tau>\tilde{\tau}_m(s)} \hspace{-0.5cm}f_{\tau}(\tau)\pi(m|\tau)m(t)(\int_{l=-\infty}^s \hspace{-0.4cm} m(l)\de l-\mu c s) \de \tau    \de m
		\\& + \frac{1}{ \bar{m}(t)}\int_{m ,\tau<\tilde{\tau}_m(s)}\hspace{-0.4cm}f_{\tau}(\tau)\pi(m|\tau)m(t) ((1-c)s-\tau) \de \tau    \de m \nonumber \\& +c \mathbb{E}(\tau|t),
		\end{align*}
		where $\bar{m}(t)=\int_{\tau,m}f_{\tau}(\tau)\pi(m|\tau)m(t)\de\tau \de m$. 
	\end{lemma}
	
	% 	\begin{proof}
	% 		See Appendix.
	% 	\end{proof}
	
	As mentioned before,  the planner restricts his attention to the set of obedient   signaling strategies. The definition of the obedience condition is  stated below.
	\begin{definition}[Obedience Condition]
		The signaling strategy $\pi$ is  obedient if all of the agents prefer to arrive at the queue at the time they are recommended to do so. That is
		\begin{align*}
		t \in \arg \min_s \bar{c}_{t,\pi}(s), \quad \forall t.
		\end{align*}
	\label{deff:obed}
	\end{definition}
	%	In order to simplify the obedience condition, we define some terminologies that will be used later on.
	%	\begin{definition}
	%		A time $t$ is called interior with respect to the signaling strategy $\pi$, if there exists an $m \in M$ for which $\pi(m|\tau)>0$ for some $\tau$, and we have $m(t)>0$ in a neighborhood of $t$. A time $t$ is called a left endpoint if for all $m \in M$ for which $\pi(m|\tau)>0$ for some $\tau$, we have $m(t)=0$ in the left neighborhood of $t$ and there exists a $m \in M$ for which $\pi(m|\tau)>0$ for some $\tau$ and we have  $m(t)>0$ in a right neighborhood of $t$. Similarly, a right endpoint is defined. 
	%	\end{definition}
	Definition \ref{deff:obed} states that $t$ must be a global minimizer of $\bar{c}_{t,\pi}(s)$ for $\pi$ to be obedient. In the next lemma, we will show that for a signaling strategy $\pi$ to be obedient, it is necessary and sufficient for $t$ to be a local minimizer of $\bar{c}_{t,\pi}(s)$.
	\begin{lemma}\label{obedcon}
		The signaling strategy $\pi$ is obedient if and only if $\frac{\de}{\de s}\bar{c}_{t,\pi}(s)|_t=0$ for all   times $t$, which implies the following must hold for an obedient signaling strategy.
		\begin{align*}
		&(1-c)\int_{m}\int_{\tau=0}^{\tilde{\tau}_m(t)}f_{\tau}(\tau)\pi(m|\tau)m(t)\de\tau \de m \nonumber \\& +\frac{1}{\mu}\int_{m} \int_{\tilde{\tau}_m(t)}^{\infty}  f_{\tau}(\tau)\pi(m|\tau)m(t)(m(t)-\mu c)\de\tau  =0.
		\end{align*}
	\end{lemma}
	% 	\begin{proof}
	% 		See Appendix.
	% 	\end{proof}
	
	% Using Lemma \ref{obedcon},	we can  simplify the obedience conditions as it is stated  in the next lemma.
	% 	\begin{lemma}
	% 		\label{inccongen}
	% 		A signalling strategy $\pi$ is obedient iff the following holds.
	% 		\begin{align}
	% 		&(1-c)\int_{m}\int_{\tau=0}^{\tilde{\tau}_m(t)}f_{\tau}(\tau)\pi(m|\tau)m(t)\de\tau \de m \nonumber \\& +\frac{1}{\mu}\int_{m} \int_{\tilde{\tau}_m(t)}^{\infty}  f_{\tau}(\tau)\pi(m|\tau)m(t)(m(t)-\mu c)\de\tau  =0 
	% 		\end{align}
	% 	\end{lemma}
	% 	\begin{proof}
	% 		See Appendix.
	% 	\end{proof}

	\section{Full Information and No Information Extremes}\label{extreme}
	In this section, we characterize the full information (all agents know the value of $\tau$) and the no information (there is no signal sent to the agents about the value of $\tau$) equilibrium arrival processes. 
	\begin{theorem}[Full Information]\label{fullinfo}
		The full information equilibrium arrival process for the service time $\tau$ is as follows. 
		\begin{align*}
		m(t)=\mu c, \quad t \in (\tau-\frac{1-c}{\mu c},\tau+\frac{1}{\mu}),
		\end{align*}
		and $m(t)=0$ elsewhere. 
	\end{theorem}
	% 	\begin{proof}
	% 		See Appendix.
	% 	\end{proof}
	Note that the full information equilibrium arrival process induces a single queue throughout the whole time horizon and the queue is cleared out at the same time the arrival process is ended.  
	
	Next, we investigate the equilibrium when the agents have no information about $\tau$, other than its prior distribution, $f_{\tau}(\cdot)$. For this part, we restrict our attention to the set of arrival processes that satisfy assumption (\ref{piece}).  Also, we assume $f_{\tau}(\cdot)$  is an exponential distribution with parameter $\lambda$. Note that this assumption is not critical in finding the no information equilibrium and one can search the equilibrium arrival process for a different $f_{\tau}(\cdot)$. As we will see in the proof of Theorem \ref{noinfo}, having a different $f_{\tau}(\cdot)$ will induce a different differential equation to be solved than the one we solve in this paper.  
	
	Before stating the equilibrium, we present the following lemma that will enable us to restrict our attention to a smaller set of arrival processes for the no information equilibrium. 
	\begin{lemma}
		For the no information equilibrium arrival process, $m$,  we  have $m(t)
		\leq \mu$ for all $t$ and $m$ can not include a delta function.
		\label{mgrmu}
	\end{lemma}
	% 	\begin{proof}
	% 		See Appendix.
	% 	\end{proof}
	Using Lemma \ref{mgrmu}, we can characterize the  no information equilibrium arrival process.
	\begin{theorem}[No Information]\label{noinfo}
		The no information equilibrium arrival process, if it exists, is as follows.
		\begin{align*}
		m(t)=\mu- \frac{\mu}{\beta-\lambda t}, \quad t \in [t_1,t_2],
		\end{align*}
		where $\beta=-\ln(1-c)+\frac{\lambda}{\mu}+1$, $t_2=\frac{-\ln(1-c)}{\lambda}+\frac{1}{\mu}$, and $t_1$ is derived from either of the following equations (or possibly both, which results in two solutions for the equilibrium).
		\begin{align*}
		&	\ln(1-c)+\lambda t_1+\ln(\frac{\lambda}{\mu}-\ln(1-c)-\lambda t_1+1 )=0,\    t_1\geq 0\\
		&t_1=\frac{1-c}{\lambda c}\ln(1-c)-\frac{1-c}{\mu c}+\frac{1}{\lambda}, \quad t_1<0.
		\end{align*}
	\end{theorem}
	% 	\begin{proof}
	% 		See Appendix.
	% 	\end{proof}

	\section{Structural Results}\label{structural}
	In this section, we assume that the planner restricts her attention to a set of arrival processes that satisfy  assumption (\ref{interval}) and for such strategies, we present a structural property in the next theorem. 

	\begin{theorem}\label{fulinfointerval}
		If a signaling strategy $\pi(\cdot|\tau)$ that satisfies assumption (\ref{interval}) is obedient and if we assume $\bar{t}_{\tau}-\underline{t}_{\tau}\leq \frac{1}{\mu c}$ and $c\leq 0.5$,
		then, $\pi(\cdot|\tau)$ is supported only over the full information equilibrium arrival process characterized in Theorem \ref{fullinfo}.
	\end{theorem}
	Note that the interval $\frac{1}{\mu c}$ is the time span of the equilibrium arrival process  in the full information case.  Theorem \ref{fulinfointerval} indicates that if the planner wants to induce a lower social cost than the full information equilibrium social cost, he should expand the time span of the arrival processes to intervals longer than $\frac{1}{\mu c}$. 
	\section{GPM Formulation}\label{GPM}
	In this section, we formulate our problem as a generalized problem of moments (GPM). A GPM is an optimization problem over finite probability measures that minimizes a cost that is linear in moments w.r.t. those measures, subject to constraints that are linear w.r.t. those moments.  The GPM formulation will allow us to utilize the computation tools available for such problems to do numerical analysis for our model. In order to express our problem as a GPM,  we  impose two assumptions (\ref{interval}) and (\ref{qpos})  on the set of arrival processes. 
	
	% 	We first give an overview of generalized problem of moments (GPM). The GPM is formulated as below \cite{lasserre2008semidefinite,henrion2009gloptipoly}.
	% 	\begin{align*}
	% 	    &\min_{\mu} \quad  \sum_k \int_{\mathbb{K}_k} g_{0k}(x)\de \mu_k(x) \\
	% 	    & s.t. \ \quad \sum_k \int_{\mathbb{K}_k} h_{jk}(x)\de \mu_k(x) \geq b_j, \ j=0,1,\cdots
	% 	\end{align*}
	% 	where $\mu_k$ is supported on $\mathbb{K}_k$, which is given as follows. 
	% 	\begin{align*}
	% 	    	\mathbb{K}_k=\{x \in \mathbb{R}^{n_k} : g_{ik}(x)\geq 0, \ i=1,2,\cdots\}
	% 	\end{align*}
	% and $g_{ik}(x)$ and $h_{jk}(x)$ are given real polynomials and $b_j$ are real constants. 
	
	% 	In order to express our problem as a GPM, we assume that the planner restricts his attention to a set of arrival processes $\tilde{M}$,  for which we have the following. First,  we make the same assumption as the one made in the previous section which is  for all $m\in \tilde{M}$ for which we have $\pi(m|\tau)>0$, we have $m(t)>0$ for $\underline{t}_{\tau}\leq t \leq \bar{t}_{\tau}$ and $m(t)=0$ for $t<\underline{t}_{\tau}$ or $t>\bar{t}_{\tau}$ for some $\underline{t}_{\tau}$ and $\bar{t}_{\tau}$. Second,  we assume that whenever $q_{\tau,m}(t)=0$ for $\underline{t}_{\tau}\leq t \leq \bar{t}_{\tau}$ and $t>\tau$, we have $m(t)=\mu$. 
	
	%%%%
	%	Third, we  assume that the queue size is zero at $\bar{t}_{\tau}$ for all $m$ with $\pi(m|\tau)>0$. 
	
	We  define $\underline{\tau}(t)$ and $\bar{\tau}(t)$ to be the inverse of $\underline{t}_{\tau}$ and $\bar{t}_{\tau}$, respectively. That is, for $\tau<\underline{\tau}(t)$ or $\tau> \bar{\tau}(t)$, we have $m(t)=0$ for all $m$ with $\pi(m|\tau)>0$.

	The obedience condition is simplified  in the next lemma. 
	\begin{lemma}
		\label{obedGPM}
		A signaling strategy $\pi$ that satisfies assumptions (\ref{interval}) and (\ref{qpos}) is obedient  iff the following holds.
		\begin{align*}
		&\int_{\underline{\tau}(t)}^{\bar{\tau}(t)}f_{\tau}(\tau)R_{m,\tau}(t,t) \de\tau =\mu c \int_{\underline{\tau}(t)}^{\bar{\tau}(t)}f_{\tau}(\tau)\bar{m}_{\tau}(t) \de\tau,
		\end{align*}
		where we denote $\bar{m}_{\tau}(t)=\int_{m}\pi(m|\tau)m(t)\de m$ and $R_{m,\tau}(t,s)=\int_{m}\pi(m|\tau)m(t)m(s) \de m$.
	\end{lemma}
	% 		\begin{proof}
	% 			See Appendix.
	% 		\end{proof}
	
	% 	and for $s>t$ we have
	% 	\begin{align}
	% 	\bar{c}'_{t}(s)=&\mathbf{1}(\underline{\tau}(s)<\bar{\tau}(t)) \frac{\int_{\underline{\tau}(s)}^{\bar{\tau}(t)}\int_{m}f_{\tau}(\tau)\pi(m|\tau)m(t)\frac{m(s)}{\mu}d\tau dm}{\int_{\underline{\tau}(t)}^{\bar{\tau}(t)}\int_{m}f_{\tau}(\tau)\pi(m|\tau)m(t)d\tau dm} \nonumber \\
	% 	&+\frac{\int_{\underline{\tau}(t)}^{\min(\underline{\tau}(s),\bar{\tau}(t))}\int_{m}f_{\tau}(\tau)\pi(m|\tau)m(t)d\tau dm}{\int_{\underline{\tau}(t)}^{\bar{\tau}(t)}\int_{m}f_{\tau}(\tau)\pi(m|\tau)m(t)d\tau dm}-c,
	% 	\end{align}
	% 	and finally, for $s<t$
	% 	\begin{align}
	% 	\bar{c}'&_{t}(s)\nonumber \\=&\mathbf{1}(\underline{\tau}(t)<\bar{\tau}(s)) \frac{\int_{\underline{\tau}(t)}^{\bar{\tau}(s)}\int_{m}f_{\tau}(\tau)\pi(m|\tau)m(t)\frac{m(s)}{\mu}d\tau dm}{\int_{\underline{\tau}(t)}^{\bar{\tau}(t)}\int_{m}f_{\tau}(\tau)\pi(m|\tau)m(t)d\tau dm}-c,
	% 	\end{align}

	% We notice that one necessary condition for the above constraints is $\bar{m}(s)\leq \bar{m}(t)$ for $s<t$. Considering the fact that $\bar{m}(t)\leq \mu$ for $t \in (\bar{t_3},t_1)$, we conclude that we must have $\bar{m}(t)\leq \mu$ for all $t\in [t_0,t_1]$.
	One can easily see that the full information signaling strategy, i.e., $\pi(m|\tau)=\mathbf{1}(m=m^F_{\tau})$, where $m^F_{\tau}$ is the full information equilibrium characterized in Theorem \ref{fullinfo}, satisfies the above obedience constraints.

	We can also simplify the planner's objective as follows. 
	\begin{lemma}\label{objGPM}
		The planner's objective is given below if he restricts his attentions to the signalling strategies that satisfy assumptions (\ref{interval}) and (\ref{qpos}).
		\begin{align*}
		\bar{s}(\pi)&=\frac{1}{\mu}\int_{\tau}f_{\tau}(\tau)(\int_{t=\underline{t}_{\tau}}^{\bar{t}_{\tau}}\int_{s=\underline{t}_{\tau}}^t(R_{m,\tau}(t,s)-\mu c\bar{m}_{\tau}(t))\de s \de t\nonumber \\&\hspace{2.2cm}+\mu c (\tau-\underline{t}_{\tau}))\de\tau
		%				
		%				\int_{t}\int_{\tau,m}f_{\tau}(\tau)\pi(m|\tau)m(t)c_{\tau,m}(t)\\
		%				&=\int_{t}\int_{\tau,m}f_{\tau}(\tau)\pi(m|\tau)m(t) \nonumber \\& \hspace{0.7cm}(\frac{q(t)}{\mu}+c(\tau-t)^++(1-c)(t-\tau)^+)\\
		%				&=\int_{t}\int_{\tau,m}f_{\tau}(\tau)\pi(m|\tau)m(t)\nonumber \\& \hspace{0.7cm}(\frac{\int_{s=0}^tm(s)ds-\mu(t-\tau)^+}{\mu}+c(\tau-t)+(t-\tau)^+)\\
		%				&=\int_{t}\int_{\tau,m}f_{\tau}(\tau)\pi(m|\tau)m(t)(\frac{\int_{s=0}^tm(s)ds}{\mu}+c(\tau-t))\\
		%				&=\frac{1}{\mu}\int_{\tau}f_{\tau}(\tau)(\int_{t=\underline{t}_{\tau}}^{\bar{t}_{\tau}}\int_{s=\underline{t}_{\tau}}^t(R_{m,\tau}(t,s)-\mu c\bar{m}_{\tau}(t))ds dt\nonumber \\&\hspace{2.2cm}+\mu c (\tau-\underline{t}_{\tau}))d\tau
		\end{align*}
	\end{lemma}
	% 		\begin{proof}
	% 			See Appendix.
	% 			\end{proof}
	According to lemmas \ref{obedGPM} and \ref{objGPM}, the planner's objective is linear in  moments of $m$ with respect to the measure $\pi(m|\tau)$. Also, the obedience condition is linear in moments of $m$. However, $m$ is supported over real numbers and therefore, the measure $\pi$ is not a finite measure. But for a problem to be a GPM, we must have finite probability measures. In order to have a finite measure, we need to discretize the time and consider a discretized version of the optimization problem, which is a GPM.   Therefore,  we can numerically solve it using the computation tools available for these types of problems such as Gloptipoly \cite{henrion2009gloptipoly}.
	In the next section, we will present these numerical results.
	
	Next lemma presents a result similar to one presented in Theorem \ref{fulinfointerval} for the signaling strategies that satisfy assumptions (\ref{interval}) and (\ref{qpos}).
	
	\begin{theorem}
		\label{fulinfores}
		If a signaling strategy $\pi(\cdot|\tau)$ that satisfies assumptions (\ref{interval}) and (\ref{qpos}) is obedient and if we assume $\bar{t}_{\tau}-\underline{t}_{\tau}\leq \frac{1}{\mu c}$ 
		%and if we set $\underline{t}_{\tau}=\tau-\frac{1-c}{c \mu}$, $\bar{t}_{\tau}=\tau+\frac{1}{\mu}$ 
		then $\pi(\cdot|\tau)$ is supported only over the full information equilibrium arrival process characterized in Theorem \ref{fullinfo}.
	\end{theorem}
	% 	\begin{proof}
	% 		See Appendix.
	% 		\end{proof}
	Note that the result of Theorem \ref{fulinfores} holds regardless of the value of $c$, while in Theorem \ref{fulinfointerval}, we must have $c\leq 0.5$ for the result to hold.

	\section{Numerical Analysis}\label{numerical}
	In this section, based on the GPM formulation of our problem, we use Gloptipoly to solve the problem numerically. In this paper, we consider uniform discretization of time. 
	
	As showed in the previous section, if we restrict our attention to the signaling strategies that satisfy assumptions (\ref{interval}) and (\ref{qpos}), and if $\bar{t}_{\tau}-\underline{t}_{\tau}\leq \frac{1}{\mu c}$, then the solution is known to have support only on the full information equilibrium of Theorem \ref{fullinfo}. This result is  numerically confirmed as it is shown in Fig. \ref{fulinfonum} and \ref{fulinfonum8} for $c=0.5$ and $c=0.8$, respectively, and for $\mu=0.5$ and a bounded discrete interval  of $\tau \in \{3, 3.5, 4, 4.5, 5, 5.5, 6\}$ with uniform distribution.

	\begin{figure*}[ht]
		\centering
		\begin{subfigure}[b]{0.495\textwidth}
			\centering
%	\begin{figure}
%		\centering
		\begin{tikzpicture}
		
		% defining custom colors
		\definecolor{mycolor1}{rgb}{0.15,0.15,0.15}
		\definecolor{mycolor2}{rgb}{0,0.447,0.741}
		\definecolor{mycolor3}{rgb}{0.85,0.325,0.098}
		\definecolor{mycolor4}{rgb}{0.929,0.694,0.125}
		\definecolor{mycolor5}{rgb}{0.494,0.184,0.556}
		\definecolor{mycolor6}{rgb}{0.466,0.674,0.188}
		\definecolor{mycolor7}{rgb}{0.301,0.745,0.933}
		\definecolor{mycolor8}{rgb}{0.635,0.078,0.184}
		
		% Axis at [0.13 0.11 0.78 0.82]
		\begin{axis}[
		scale only axis,
		every outer x axis line/.append style={mycolor1},
		every x tick label/.append style={font=\color{mycolor1}},
		every outer y axis line/.append style={mycolor1},
		every y tick label/.append style={font=\color{mycolor1}},
		width=2.42778in,
		height=1.85417in,
		y label style={at={(axis description cs:-0.07,0.9)},rotate=270,anchor=south},
		x label style={at={(axis description  cs:0.98,-0.01)},anchor=north},
		xlabel=$t$,
		ylabel=$m(t)$,
		ytick={0.05,0.15,0.25},
		yticklabels={0.05,0.15,0.25},
		xtick={2,4,6,8},
		xticklabels={2,4,6,8},
		xmin=0, xmax=9,
		ymin=0, ymax=0.32,
		axis on top]
		
		\addplot [
		color=mycolor2,
		solid,
		mark=asterisk,
		mark options={solid}
		]
		coordinates{
			(1.25,0)
			(1.5,0.249309)
			(1.75,0.250024)
			(2,0.250129)
			(2.25,0.250076)
			(2.5,0.250058)
			(2.75,0.250044)
			(3,0.25004)
			(3.25,0.250032)
			(3.5,0.250032)
			(3.75,0.250025)
			(4,0.250028)
			(4.25,0.250021)
			(4.5,0.250052)
			(4.75,0.250048)
			(5,0.250045)
			(5.25,0.250039)
			(5.5,0)
			
		};
		
		\addplot [
		color=mycolor3,
		solid
		]
		coordinates{
			(1.75,0)
			(2,0.249867)
			(2.25,0.24994)
			(2.5,0.249974)
			(2.75,0.249987)
			(3,0.249994)
			(3.25,0.250001)
			(3.5,0.250001)
			(3.75,0.250007)
			(4,0.250004)
			(4.25,0.250011)
			(4.5,0.250029)
			(4.75,0.25003)
			(5,0.250031)
			(5.25,0.250036)
			(5.5,0.250047)
			(5.75,0.250041)
			(6,0)
			
		};
		
		\addplot [
		color=mycolor4,
		solid
		]
		coordinates{
			(2.25,0)
			(2.5,0.24999)
			(2.75,0.249993)
			(3,0.249994)
			(3.25,0.249995)
			(3.5,0.249997)
			(3.75,0.249997)
			(4,0.249998)
			(4.25,0.249997)
			(4.5,0.249997)
			(4.75,0.249996)
			(5,0.249995)
			(5.25,0.24999)
			(5.5,0.249998)
			(5.75,0.249994)
			(6,0.250036)
			(6.25,0.250031)
			(6.5,0)
			
		};
		
		\addplot [
		color=mycolor5,
		solid
		]
		coordinates{
			(2.75,0)
			(3,0.249997)
			(3.25,0.249998)
			(3.5,0.249993)
			(3.75,0.249993)
			(4,0.249994)
			(4.25,0.249995)
			(4.5,0.249992)
			(4.75,0.249993)
			(5,0.249994)
			(5.25,0.249996)
			(5.5,0.249998)
			(5.75,0.250002)
			(6,0.25)
			(6.25,0.249999)
			(6.5,0.25003)
			(6.75,0.250025)
			(7,0)
			
		};
		
		\addplot [
		color=mycolor6,
		solid
		]
		coordinates{
			(3.25,0)
			(3.5,0.250004)
			(3.75,0.250004)
			(4,0.249994)
			(4.25,0.249994)
			(4.5,0.249988)
			(4.75,0.249989)
			(5,0.24999)
			(5.25,0.249992)
			(5.5,0.249995)
			(5.75,0.249998)
			(6,0.249998)
			(6.25,0.250001)
			(6.5,0.250001)
			(6.75,0.250001)
			(7,0.250026)
			(7.25,0.250022)
			(7.5,0)
			
		};
		
		\addplot [
		color=mycolor7,
		solid
		]
		coordinates{
			(3.75,0)
			(4,0.250009)
			(4.25,0.250009)
			(4.5,0.249987)
			(4.75,0.249987)
			(5,0.249988)
			(5.25,0.24999)
			(5.5,0.249994)
			(5.75,0.249996)
			(6,0.249996)
			(6.25,0.249998)
			(6.5,0.249998)
			(6.75,0.250001)
			(7,0.250001)
			(7.25,0.250002)
			(7.5,0.250023)
			(7.75,0.250021)
			(8,0)
			
		};
		
		\addplot [
		color=mycolor8,
		solid
		]
		coordinates{
			(4.25,0)
			(4.5,0.249983)
			(4.75,0.249984)
			(5,0.249985)
			(5.25,0.249986)
			(5.5,0.249996)
			(5.75,0.249998)
			(6,0.249997)
			(6.25,0.249999)
			(6.5,0.249998)
			(6.75,0.250001)
			(7,0.250001)
			(7.25,0.250004)
			(7.5,0.250005)
			(7.75,0.250008)
			(8,0.250028)
			(8.25,0.250028)
			(8.5,0)
			
		};
		
		\end{axis}
		
		\end{tikzpicture}
		\caption{$c=0.5$ and $\bar{t}_{\tau}-\underline{t}_{\tau}= \frac{1}{\mu c}$}
		\label{fulinfonum}
		\end{subfigure}
	\begin{subfigure}[b]{0.495\textwidth}
		\centering
		
		\begin{tikzpicture}
		
		% defining custom colors
		\definecolor{mycolor1}{rgb}{0.15,0.15,0.15}
		\definecolor{mycolor2}{rgb}{0,0.447,0.741}
		\definecolor{mycolor3}{rgb}{0.85,0.325,0.098}
		\definecolor{mycolor4}{rgb}{0.929,0.694,0.125}
		\definecolor{mycolor5}{rgb}{0.494,0.184,0.556}
		\definecolor{mycolor6}{rgb}{0.466,0.674,0.188}
		\definecolor{mycolor7}{rgb}{0.301,0.745,0.933}
		\definecolor{mycolor8}{rgb}{0.635,0.078,0.184}
		
		% Axis at [0.13 0.11 0.78 0.82]
		\begin{axis}[
		scale only axis,
		every outer x axis line/.append style={mycolor1},
		every x tick label/.append style={font=\color{mycolor1}},
		every outer y axis line/.append style={mycolor1},
		every y tick label/.append style={font=\color{mycolor1}},
		width=2.42778in,
		height=1.85417in,
		y label style={at={(axis description cs:-0.07,0.9)},rotate=270,anchor=south},
		x label style={at={(axis description  cs:0.98,-0.01)},anchor=north},
		xlabel=$t$,
		ylabel=$m(t)$,
		ytick={0,0.2,0.4},
		yticklabels={0,0.2,0.4},
		xtick={2,4,6,8},
		xticklabels={2,4,6,8},
		xmin=1, xmax=9,
		ymin=0, ymax=0.5,
		axis on top]
		
		\addplot [
		color=mycolor2,
		solid,
		mark=asterisk,
		mark options={solid}
		]
		coordinates{
			(2.75,0)
			(3,0.399676)
			(3.25,0.400014)
			(3.5,0.400078)
			(3.75,0.400047)
			(4,0.400039)
			(4.25,0.400032)
			(4.5,0.400031)
			(4.75,0.400028)
			(5,0.400029)
			(5.25,0.400026)
			(5.5,0)
			
		};
		
		\addplot [
		color=mycolor3,
		solid
		]
		coordinates{
			(3.25,0)
			(3.5,0.399917)
			(3.75,0.399962)
			(4,0.400009)
			(4.25,0.40001)
			(4.5,0.400021)
			(4.75,0.400019)
			(5,0.400022)
			(5.25,0.40002)
			(5.5,0.400011)
			(5.75,0.400009)
			(6,0)
			
		};
		
		\addplot [
		color=mycolor4,
		solid
		]
		coordinates{
			(3.75,0)
			(4,0.399965)
			(4.25,0.399973)
			(4.5,0.399988)
			(4.75,0.399995)
			(5,0.399998)
			(5.25,0.400002)
			(5.5,0.400002)
			(5.75,0.400007)
			(6,0.400041)
			(6.25,0.400029)
			(6.5,0)
			
		};
		
		\addplot [
		color=mycolor5,
		solid
		]
		coordinates{
			(4.25,0)
			(4.5,0.399976)
			(4.75,0.399975)
			(5,0.399984)
			(5.25,0.39999)
			(5.5,0.399997)
			(5.75,0.400002)
			(6,0.40001)
			(6.25,0.40001)
			(6.5,0.400032)
			(6.75,0.400023)
			(7,0)
			
		};
		
		\addplot [
		color=mycolor6,
		solid
		]
		coordinates{
			(4.75,0)
			(5,0.399983)
			(5.25,0.399979)
			(5.5,0.399997)
			(5.75,0.399998)
			(6,0.399992)
			(6.25,0.399998)
			(6.5,0.400004)
			(6.75,0.400006)
			(7,0.400025)
			(7.25,0.400018)
			(7.5,0)
			
		};
		
		\addplot [
		color=mycolor7,
		solid
		]
		coordinates{
			(5.25,0)
			(5.5,0.40001)
			(5.75,0.400002)
			(6,0.399983)
			(6.25,0.39999)
			(6.5,0.399988)
			(6.75,0.399995)
			(7,0.399997)
			(7.25,0.400003)
			(7.5,0.400017)
			(7.75,0.400015)
			(8,0)
			
		};
		
		\addplot [
		color=mycolor8,
		solid
		]
		coordinates{
			(5.75,0)
			(6,0.399991)
			(6.25,0.399991)
			(6.5,0.399993)
			(6.75,0.399994)
			(7,0.399996)
			(7.25,0.399997)
			(7.5,0.400001)
			(7.75,0.400003)
			(8,0.400018)
			(8.25,0.400018)
			(8.5,0)
			
		};
		
		\end{axis}
		
		\end{tikzpicture}
		
		\caption {$c=0.8$ and $\bar{t}_{\tau}-\underline{t}_{\tau}= \frac{1}{\mu c}$.}
		\label{fulinfonum8}
	\end{subfigure}

		\begin{subfigure}[b]{0.49\textwidth}
		\centering
		% This file was created by matlab2tikz v0.0.5.
		% Copyright (c) 2008--2010, Nico Schlömer <nico.schloemer@ua.ac.be>
		% All rights reserved.
		%
		% The latest updates can be retrieved from
		%  http://win.ua.ac.be/~nschloe/content/matlab2tikz/
		% and
		%  http://www.mathworks.com/matlabcentral/fileexchange/22022 .
		% where you can also make suggestions and rate matlab2tikz.
		
		\begin{tikzpicture}
		
		% defining custom colors
		\definecolor{mycolor1}{rgb}{0.15,0.15,0.15}
		\definecolor{mycolor2}{rgb}{0,0.447,0.741}
		\definecolor{mycolor3}{rgb}{0.85,0.325,0.098}
		\definecolor{mycolor4}{rgb}{0.929,0.694,0.125}
		\definecolor{mycolor5}{rgb}{0.494,0.184,0.556}
		\definecolor{mycolor6}{rgb}{0.466,0.674,0.188}
		\definecolor{mycolor7}{rgb}{0.301,0.745,0.933}
		\definecolor{mycolor8}{rgb}{0.635,0.078,0.184}
		
		% Axis at [0.13 0.11 0.78 0.82]
		\begin{axis}[
		scale only axis,
		every outer x axis line/.append style={mycolor1},
		every x tick label/.append style={font=\color{mycolor1}},
		every outer y axis line/.append style={mycolor1},
		every y tick label/.append style={font=\color{mycolor1}},
		width=2.42778in,
		height=1.85417in,
		y label style={at={(axis description cs:-0.07,0.9)},rotate=270,anchor=south},
		x label style={at={(axis description  cs:0.98,-0.01)},anchor=north},
		xlabel=$t$,
		ylabel=$m(t)$,
		ytick={0.05,0.15,0.25},
		yticklabels={0.05,0.15,0.25},
		xtick={2,4,6,8},
		xticklabels={2,4,6,8},
		xmin=0, xmax=9,
		ymin=0, ymax=0.35,
		axis on top]
		
		\addplot [
		color=mycolor2,
		solid,
			mark=asterisk,
		mark options={solid}
		]
		coordinates{
			(0.25,8.0936e-09)
			(0.5,1.23804e-08)
			(0.75,1.76539e-08)
			(1,3.85679e-08)
			(1.25,1.29277e-07)
			(1.5,2.79766e-05)
			(1.75,0.293776)
			(2,0.315362)
			(2.25,0.323896)
			(2.5,0.323595)
			(2.75,0.326517)
			(3,0.3159)
			(3.25,0.320562)
			(3.5,0.311609)
			(3.75,0.304057)
			(4,0.297888)
			(4.25,0.292852)
			(4.5,0.288706)
			(4.75,0.28525)
			(5,0)
			
		};
		
		\addplot [
		color=mycolor3,
		solid
		]
		coordinates{
			(0.5,0)
			(0.75,1.35783e-08)
			(1,2.00154e-08)
			(1.25,2.98783e-08)
			(1.5,5.26445e-08)
			(1.75,0.072485)
			(2,0.226945)
			(2.25,0.267752)
			(2.5,0.283321)
			(2.75,0.294628)
			(3,0.291133)
			(3.25,0.299063)
			(3.5,0.293812)
			(3.75,0.288979)
			(4,0.284855)
			(4.25,0.281401)
			(4.5,0.278508)
			(4.75,0.276069)
			(5,0.281816)
			(5.25,0.279232)
			(5.5,0)
			
		};
		
		\addplot [
		color=mycolor4,
		solid
		]
		coordinates{
			(1,0)
			(1.25,1.70762e-08)
			(1.5,2.14512e-08)
			(1.75,5.96703e-08)
			(2,0.109147)
			(2.25,0.192984)
			(2.5,0.229702)
			(2.75,0.25218)
			(3,0.258164)
			(3.25,0.270447)
			(3.5,0.270121)
			(3.75,0.268905)
			(4,0.267503)
			(4.25,0.266152)
			(4.5,0.264923)
			(4.75,0.263831)
			(5,0.270108)
			(5.25,0.268555)
			(5.5,0.274552)
			(5.75,0.272727)
			(6,0)
			
		};
		
		\addplot [
		color=mycolor5,
		solid
		]
		coordinates{
			(1.5,0)
			(1.75,2.23226e-08)
			(2,2.33574e-07)
			(2.25,0.108542)
			(2.5,0.169162)
			(2.75,0.204261)
			(3,0.220949)
			(3.25,0.238146)
			(3.5,0.243381)
			(3.75,0.246248)
			(4,0.247917)
			(4.25,0.248938)
			(4.5,0.249588)
			(4.75,0.250014)
			(5,0.256885)
			(5.25,0.25649)
			(5.5,0.262897)
			(5.75,0.261992)
			(6,0.267909)
			(6.25,0.26668)
			(6.5,0)
			
		};
		
		\addplot [
		color=mycolor6,
		solid
		]
		coordinates{
			(2,0)
			(2.25,0.0096689)
			(2.5,0.0982827)
			(2.75,0.148162)
			(3,0.177387)
			(3.25,0.200339)
			(3.5,0.212084)
			(3.75,0.219729)
			(4,0.224992)
			(4.25,0.22879)
			(4.5,0.231638)
			(4.75,0.233842)
			(5,0.241407)
			(5.25,0.242366)
			(5.5,0.249247)
			(5.75,0.249415)
			(6,0.255682)
			(6.25,0.255331)
			(6.5,0.261178)
			(6.75,0.260458)
			(7,0)
			
		};
		
		\addplot [
		color=mycolor7,
		solid
		]
		coordinates{
			(2.5,0)
			(2.75,0.0878424)
			(3,0.130555)
			(3.25,0.159698)
			(3.5,0.178442)
			(3.75,0.191225)
			(4,0.200352)
			(4.25,0.207134)
			(4.5,0.212345)
			(4.75,0.216458)
			(5,0.224771)
			(5.25,0.227185)
			(5.5,0.234574)
			(5.75,0.235894)
			(6,0.242533)
			(6.25,0.243121)
			(6.5,0.249244)
			(6.75,0.249317)
			(7,0.254807)
			(7.25,0.254501)
			(7.5,0)
			
		};
		
		\addplot [
		color=mycolor8,
		solid
		]
		coordinates{
			(3,0)
			(3.25,0.12339)
			(3.5,0.148396)
			(3.75,0.165772)
			(4,0.178351)
			(4.25,0.187799)
			(4.5,0.195118)
			(4.75,0.200937)
			(5,0.209917)
			(5.25,0.21363)
			(5.5,0.221472)
			(5.75,0.22382)
			(6,0.230792)
			(6.25,0.232216)
			(6.5,0.238583)
			(6.75,0.23936)
			(7,0.24503)
			(7.25,0.24536)
			(7.5,0.250028)
			(7.75,0.250028)
			(8,0)
			
		};
		
		\end{axis}
		
		\end{tikzpicture}
		\caption{ $c=0.5$ and $\bar{t}_{\tau}-\underline{t}_{\tau}>\frac{1}{\mu c}$.}
		\label{longernum}
	\end{subfigure}
		\begin{subfigure}[b]{0.49\textwidth}
		\centering
		\begin{tikzpicture}
		
		% defining custom colors
		\definecolor{mycolor1}{rgb}{0.15,0.15,0.15}
		\definecolor{mycolor2}{rgb}{0,0.447,0.741}
		\definecolor{mycolor3}{rgb}{0.85,0.325,0.098}
		\definecolor{mycolor4}{rgb}{0.929,0.694,0.125}
		\definecolor{mycolor5}{rgb}{0.494,0.184,0.556}
		\definecolor{mycolor6}{rgb}{0.466,0.674,0.188}
		\definecolor{mycolor7}{rgb}{0.301,0.745,0.933}
		\definecolor{mycolor8}{rgb}{0.635,0.078,0.184}
		
		% Axis at [0.13 0.11 0.78 0.82]
		\begin{axis}[
		scale only axis,
		every outer x axis line/.append style={mycolor1},
		every x tick label/.append style={font=\color{mycolor1}},
		every outer y axis line/.append style={mycolor1},
		every y tick label/.append style={font=\color{mycolor1}},
		width=2.42778in,
		height=1.85417in,
		y label style={at={(axis description cs:-0.07,0.9)},rotate=270,anchor=south},
		x label style={at={(axis description  cs:0.98,-0.01)},anchor=north},
		xlabel=$t$,
		ylabel=$m(t)$,
		ytick={0.1,0.3,0.5},
		yticklabels={0.1,0.3,0.5},
		xtick={2,4,6,8},
		xticklabels={2,4,6,8},
		xmin=2, xmax=9,
		ymin=0, ymax=0.65,
		axis on top]
		
		\addplot [
		color=mycolor2,
		solid,
			mark=asterisk,
		mark options={solid}
		]
		coordinates{
			(0.25,0)
			(0.5,0)
			(0.75,0)
			(1,0)
			(1.25,0)
			(1.5,0)
			(1.75,0)
			(2,0)
			(2.25,-7.15168e-10)
			(2.5,2.75855e-10)
			(2.75,1.90336e-09)
			(3,1.06573e-08)
			(3.25,4.27066e-07)
			(3.5,0.399982)
			(3.75,0.474645)
			(4,0.50261)
			(4.25,0.514593)
			(4.5,0.52422)
			(4.75,0.531069)
			(5,0.521902)
			(5.25,0.530977)
			(5.5,0)
			(5.75,0)
			(6,0)
			(6.25,0)
			(6.5,0)
			(6.75,0)
			(7,0)
			(7.25,0)
			(7.5,0)
			(7.75,0)
			(8,0)
			(8.25,0)
			(8.5,0)
			(8.75,0)
			
		};
		
		\addplot [
		color=mycolor3,
		solid
		]
		coordinates{
			(0.25,0)
			(0.5,0)
			(0.75,0)
			(1,0)
			(1.25,0)
			(1.5,0)
			(1.75,0)
			(2,0)
			(2.25,0)
			(2.5,0)
			(2.75,3.26194e-10)
			(3,6.48173e-10)
			(3.25,4.38749e-09)
			(3.5,2.62392e-08)
			(3.75,0.275407)
			(4,0.376946)
			(4.25,0.423075)
			(4.5,0.451636)
			(4.75,0.470833)
			(5,0.472573)
			(5.25,0.48717)
			(5.5,0.52813)
			(5.75,0.51423)
			(6,0)
			(6.25,0)
			(6.5,0)
			(6.75,0)
			(7,0)
			(7.25,0)
			(7.5,0)
			(7.75,0)
			(8,0)
			(8.25,0)
			(8.5,0)
			(8.75,0)
			
		};
		
		\addplot [
		color=mycolor4,
		solid
		]
		coordinates{
			(0.25,0)
			(0.5,0)
			(0.75,0)
			(1,0)
			(1.25,0)
			(1.5,0)
			(1.75,0)
			(2,0)
			(2.25,0)
			(2.5,0)
			(2.75,0)
			(3,0)
			(3.25,5.16543e-10)
			(3.5,2.37525e-09)
			(3.75,0.00279286)
			(4,0.205016)
			(4.25,0.297876)
			(4.5,0.352344)
			(4.75,0.388433)
			(5,0.405092)
			(5.25,0.42724)
			(5.5,0.468948)
			(5.75,0.464329)
			(6,0.499982)
			(6.25,0.487946)
			(6.5,0)
			(6.75,0)
			(7,0)
			(7.25,0)
			(7.5,0)
			(7.75,0)
			(8,0)
			(8.25,0)
			(8.5,0)
			(8.75,0)
			
		};
		
		\addplot [
		color=mycolor5,
		solid
		]
		coordinates{
			(0.25,0)
			(0.5,0)
			(0.75,0)
			(1,0)
			(1.25,0)
			(1.5,0)
			(1.75,0)
			(2,0)
			(2.25,0)
			(2.5,0)
			(2.75,0)
			(3,0)
			(3.25,0)
			(3.5,0)
			(3.75,2.7882e-09)
			(4,0.0073885)
			(4.25,0.153973)
			(4.5,0.238223)
			(4.75,0.293731)
			(5,0.327538)
			(5.25,0.358364)
			(5.5,0.400929)
			(5.75,0.406975)
			(6,0.442991)
			(6.25,0.439294)
			(6.5,0.469367)
			(6.75,0.461225)
			(7,0)
			(7.25,0)
			(7.5,0)
			(7.75,0)
			(8,0)
			(8.25,0)
			(8.5,0)
			(8.75,0)
			
		};
		
		\addplot [
		color=mycolor6,
		solid
		]
		coordinates{
			(0.25,0)
			(0.5,0)
			(0.75,0)
			(1,0)
			(1.25,0)
			(1.5,0)
			(1.75,0)
			(2,0)
			(2.25,0)
			(2.5,0)
			(2.75,0)
			(3,0)
			(3.25,0)
			(3.5,0)
			(3.75,0)
			(4,0)
			(4.25,0.00105803)
			(4.5,0.116948)
			(4.75,0.193094)
			(5,0.245126)
			(5.25,0.285176)
			(5.5,0.328653)
			(5.75,0.346031)
			(6,0.382431)
			(6.25,0.387594)
			(6.5,0.41801)
			(6.75,0.416656)
			(7,0.441911)
			(7.25,0.437313)
			(7.5,0)
			(7.75,0)
			(8,0)
			(8.25,0)
			(8.5,0)
			(8.75,0)
			
		};
		
		\addplot [
		color=mycolor7,
		solid
		]
		coordinates{
			(0.25,0)
			(0.5,0)
			(0.75,0)
			(1,0)
			(1.25,0)
			(1.5,0)
			(1.75,0)
			(2,0)
			(2.25,0)
			(2.5,0)
			(2.75,0)
			(3,0)
			(3.25,0)
			(3.5,0)
			(3.75,0)
			(4,0)
			(4.25,0)
			(4.5,0)
			(4.75,0.0872637)
			(5,0.158463)
			(5.25,0.208213)
			(5.5,0.25265)
			(5.75,0.281946)
			(6,0.31875)
			(6.25,0.33323)
			(6.5,0.364004)
			(6.75,0.369788)
			(7,0.39547)
			(7.25,0.396361)
			(7.5,0.417831)
			(7.75,0.416029)
			(8,0)
			(8.25,0)
			(8.5,0)
			(8.75,0)
			
		};
		
		\addplot [
		color=mycolor8,
		solid
		]
		coordinates{
			(0.25,0)
			(0.5,0)
			(0.75,0)
			(1,0)
			(1.25,0)
			(1.5,0)
			(1.75,0)
			(2,0)
			(2.25,0)
			(2.5,0)
			(2.75,0)
			(3,0)
			(3.25,0)
			(3.5,0)
			(3.75,0)
			(4,0)
			(4.25,0)
			(4.5,0)
			(4.75,0)
			(5,0)
			(5.25,0.137113)
			(5.5,0.182439)
			(5.75,0.222745)
			(6,0.259924)
			(6.25,0.283011)
			(6.5,0.314117)
			(6.75,0.326494)
			(7,0.352569)
			(7.25,0.358531)
			(7.5,0.380435)
			(7.75,0.382588)
			(8,0.400018)
			(8.25,0.400018)
			(8.5,0)
			(8.75,0)
			
		};
		
		\end{axis}
		
		\end{tikzpicture}
		\caption{ $c=0.8$ and $\bar{t}_{\tau}-\underline{t}_{\tau}>\frac{1}{\mu c}$.}
		\label{longernum8}
	\end{subfigure}
	
		\caption{$m(t)$ for different values of $\tau \in \{3, 3.5, 4, 4.5, 5, 5.5, 6\}$. The stared plots corresponds to $\tau=3.5$.}
	\end{figure*}

	In order to investigate solutions other than the full information equilibrium, we allow the interval of $\bar{t}_{\tau}-\underline{t}_{\tau}$ to be longer than $\frac{1}{\mu c}$. We set $\bar{t}_{\tau}-\underline{t}_{\tau}=\frac{1}{\mu c}+0.75$. The optimal signaling strategy for each $\tau$ turns out to have support on a singleton arrival process and the different arrival processes corresponding to each $\tau$ are represented in Fig. \ref{longernum} and \ref{longernum8} for $c=0.5$ and $c=0.8$, respectively.

	%In section \ref{structural}, we proved that if the information designer 
	%	This assumption implies that $m(t)=0$ for $t> \tau+\frac{1}{\mu}$. In order to have zero agents in the queue when the arrival process ends, we assume that $m(t)>0$ for $\tau\leq t\leq \tau+\frac{1}{\mu}$. These assumptions imply that to have arrival processes over longer intervals, we have to extend the full information equilibrium arrival process from left (before $\tau$). In the following, we have presented the plots of arrival processes under different probability distributions for $\tau$. 
	%%	
	%	We try to make a comparison between the arrival processes and the two known arrival process. First, the optimal arrival process regardless of the incentive constraints. This arrival process is $m^*(t)=\mu$ for $\tau\leq t \leq \tau+\frac{1}{\mu}$. This arrival process results in the lowest social cost for $\tau$, but of course, it does not satisfy the incentive constraints. The second arrival process is the full information arrival process. 
	%	

	An intuitive explanation about why the solution looks like what we see in Fig. \ref{longernum}, is that the planner decides to put smaller values of $\tau$ in higher priority compared to larger values. We can see that the arrival processes associated with smaller values of $\tau$ result in smaller social cost. 
	However, they do not satisfy the obedience condition and are indeed far from it.  This has been compensated with the arrival processes associated with larger values of $\tau$ that result in higher social cost but help with the obedience condition. 

	\section{Conclusion}\label{conclusion}
	In this paper, we formulated and studied an information design problem for a non-atomic service scheduling game. We characterized the two extremes of full information and the no information equilibrium and investigated the conditions in which the planner should reveal the full information to the agents. We also formulated the information design problem as a GPM by imposing some assumptions on the model and then numerically solved some examples of the problem.

	\appendix
\optv{arxiv}{
	\subsection{Proof of Lemma \ref{avcost}}
		\begin{align*}
		\bar{c}&_{t,\pi}(s)=\mathbb{E}\{c_{\tau,m}(s)|t,\pi\}\\&= \frac{\mathbb{E}\{q(s)|t,\pi\}}{\mu}+\mathbb{E}\{(s-\tau)^+-c(s-\tau)\}
		\\&= \frac{1}{\mu \bar{m}(t)}\int_{m,\tau>\tilde{\tau}_m(s)}f_{\tau}(\tau)\pi(m|\tau)m(t)\nonumber \\&(\int_{l=-\infty}^s m(l)\de l-\mu(s-\tau)^+ +\mu c(\tau-s)+\mu (s-\tau)^+) \nonumber \\& \hspace{6cm} \de \tau    \de m \nonumber \\ 
		&+ \frac{1}{ \bar{m}(t)}\int_{m,\tau<\tilde{\tau}_m(s)}f_{\tau}(\tau)\pi(m|\tau)m(t) \nonumber \\& \hspace{3cm}(c(\tau-s)+ (s-\tau)^+) \ \de \tau    \de m 
		\\=&\frac{1}{\mu \bar{m}(t)}\int_{m,\tau>\tilde{\tau}_m(s)}f_{\tau}(\tau)\pi(m|\tau)m(t)\nonumber \\& \hspace{2.2cm}(\int_{l=-\infty}^s m(l)\de l-\mu c s +\mu c\tau) \ \de \tau    \de m  \nonumber 
		\\&+ \frac{1}{ \bar{m}(t)}\int_{m,\tau<\tilde{\tau}_m(s)}f_{\tau}(\tau)\pi(m|\tau)m(t) \nonumber \\ & \hspace{3.2cm}(c(\tau-s)+(s-\tau)) \ \de \tau    \de m 
		\\ =&\frac{1}{\mu \bar{m}(t)}\int_{m,\tau>\tilde{\tau}_m(s)}f_{\tau}(\tau)\pi(m|\tau)m(t)\nonumber \\ & \hspace{3.5cm} (\int_{l=-\infty}^s m(l)\de l-\mu c s) \de \tau    \de m  \nonumber 
		\\& + \frac{1}{ \bar{m}(t)}\int_{m,\tau<\tilde{\tau}_m(s)}f_{\tau}(\tau)\pi(m|\tau)m(t) \nonumber \\ & \hspace{2.6cm} ((1-c)s-\tau) \ \de  \tau    \de m +c \mathbb{E}(\tau|t)
		\end{align*}
		where $\bar{m}(t)=\int_{\tau,m}f_{\tau}(\tau)\pi(m|\tau)m(t)\de\tau \de m$.

}
	
\subsection{Proof of Lemma \ref{obedcon}}
		In order for an agent to obey her suggestion, we must have $\bar{c}_{t,\pi}(t)$ to be the global minimizer of $\bar{c}_{t,\pi}(s)$. In the next lemma, we will show that $\bar{c}_{t,\pi}(s)$ is convex and therefore, any local minimizer is a global minimizer.
		\begin{lemma}
			$\bar{c}_{t,\pi}(s)$ is convex with respect to $s$.
			\label{convex}
		\end{lemma}
		\begin{proof}
			In order to prove convexity of $\bar{c}_{t,\pi}(s)$, we prove that its  derivative is increasing. But we first go over some preliminary results. 
			
			For $t\geq 0$, if $\int_{-\infty}^t m(s)\de s\leq \mu t$, we have
			\begin{align*}
			\int_{-\infty}^t m(s)\de s=\mu(t-\tilde{\tau}_m(t))
			\end{align*}
			Since we have $\tilde{\tau}_m(t)=0$ for  $t\leq0$, the following holds.
			\begin{align*}
			\tilde{\tau}_m(t)=(t-\int_{-\infty}^{t}\frac{m(s)}{\mu}\de s)^+
			\end{align*}
			\begin{lemma}
				\label{taucont}
				If $m(t)\leq \mu$ for all $t$, then $\tilde{\tau}_m(t)$ is continuous and increasing with respect to $t$. %and it is given by the following equation.
				% 			\begin{align}
				% 			\tilde{\tau}_m(t)=t-\int_{-\infty}^{t}\frac{m(s)}{\mu}\de s,
				% 			\end{align}
				% 		
				% 			for $t\geq 0$ and  $\tilde{\tau}_m(t)=0$ for $t<0$. 
				Furthermore, $\tilde{\tau}_m(t)$ is differentiable for all $t$ except possibly  for  $t=\tilde{t}$, where $\tilde{t}$ will be characterized in the proof.
			\end{lemma}
			\begin{proof}
				Since $m(t)\leq \mu$, there exists a time $\tilde{t}$, for which we have  $t-\int_{-\infty}^{t}\frac{m(s)}{\mu}\de s\geq 0$ for $t\geq \tilde{t}$ and $t-\int_{-\infty}^{t}\frac{m(s)}{\mu}\de s<0$ for $t< \tilde{t}$. It is clear that $\tilde{\tau}_m(t)$ is continuous, differentiable and increasing for $t< \tilde{t}$ and for $t>\tilde{t}$. Also, since the assumption of $m(t)\leq \mu$
				eliminates the possibility of $m(t)$ including a delta function, $\tilde{\tau}_m(t)$ is continuous for all $t$.
				%$\int_0^{(t)^+}(1-\frac{m(s)}{\mu})ds\geq 0$ and therefore $ \tilde{\tau}_m(t)=(\int_0^{(t)^+}(1-\frac{m(s)}{\mu})ds)^+=\int_0^{(t)^+}(1-\frac{m(s)}{\mu})ds$.
				
				% 			To prove continuity, we notice that  $\tilde{\tau}_m(t)=\int_0^{(t)^+}(1-\frac{m(s)}{\mu})ds$ and since the assumption of $m(t)\leq \mu$
				% 			eliminates the possibility of $m(t)$ including a delta function, $\tilde{\tau}_m(t)$ is continuous with respect to  $t$. It is also evident that $\tilde{\tau}_m(t)$ is increasing with respect to $t$ if $m(t)\leq \mu$. It is evident that $\tilde{\tau}_m(t)$ is differentiable for $t>0$. We have $\tilde{\tau}_m(t)=0$ for $t<0$ and therefore, it is differentiable for $t<0$. 
			\end{proof}
			Since we have assumed $m(t)\leq \mu$, and according to Lemma \ref{taucont}, we know $\tilde{\tau}_m(t)$ is continuous,  increasing, and differentiable for $t\neq \tilde{t}$,  we can write the following for $\frac{\de}{\de s}\bar{c}_{t,\pi}(s)$ for $s\neq \tilde{t}$.
			\begin{align}
			&\frac{\de}{\de s}\bar{c}_{t,\pi}(s)
			=\frac{\de}{\de s}\frac{1}{ \bar{m}(t)}\int_{\tau,m}c_{m,\tau}(s)f_{\tau}(\tau)\pi(m|\tau)m(t) \de \tau \de m\nonumber 
			\\&= \frac{\de}{\de s}\frac{1}{ \bar{m}(t)}(\int_m\int_0^{\tilde{\tau}_m(s)}\hspace{-0.35cm}(1-c)(s-\tau)f_{\tau}(\tau)\pi(m|\tau)m(t) \de \tau \de m \nonumber 
			\\&+\int_m \int_{\tilde{\tau}_m(s)}^{\infty}  (\int_{l=-\infty}^s\frac{m(l)}{\mu}\de l-c(s-\tau)) \nonumber 
			\\& \hspace{4.5cm}f_{\tau}(\tau)\pi(m|\tau)m(t) \de \tau \de m)\nonumber
			\\	&=\frac{1}{ \bar{m}(t)}(\int_m \int_0^{\tilde{\tau}_m(s)}(1-c)f_{\tau}(\tau)\pi(m|\tau)m(t) \de \tau \de m \nonumber
			+\\& \int_m\tilde{\tau}_m'(s)(1-c)(s-\tilde{\tau}_m(s))f_{\tau}(\tilde{\tau}_m(s))\pi(m|\tilde{\tau}_m(s))m(t) \de m \nonumber 
			\\	&-\int_m\tilde{\tau}_m'(s)(\int_{l=-\infty}^s\frac{m(l)}{\mu}\de l-c(s-\tilde{\tau}_m(s))) \nonumber
			\\& \hspace{4cm}f_{\tau}(\tilde{\tau}_m(s))\pi(m|\tilde{\tau}_m(s))m(t) \de m \nonumber \\&+\int_m\int_{\tilde{\tau}_m(s)}^{\infty}  (\frac{m(s)}{\mu}-c)f_{\tau}(\tau)\pi(m|\tau)m(t) \de \tau \de m)\nonumber 
			\\
			&= \frac{1}{\bar{m}(t)}( \int_m\int_{\tau=0}^{\tilde{\tau}_m(s)}f_{\tau}(\tau)\pi(m|\tau)m(t)\de\tau \de m \nonumber 
			\\& \quad  +\int_m \int_{\tilde{\tau}_m(s)}^{\infty} f_{\tau}(\tau)\pi(m|\tau)m(t)\frac{m(s)}{\mu}\de\tau \de m)-c \label{cprime}
			\end{align}
			One can see that the left and right derivative of $\bar{c}_{t,\pi}(s)$ at $s=\tilde{t}$ are equal to equation \eqref{cprime} and therefore, equation \eqref{cprime} holds for all $s$.
			According to equation \eqref{cprime}, 
			since $m(t)\leq \mu$ for all $t$, the term in the first integral of $\frac{\de}{\de s}\bar{c}_{t,\pi}(s)$, i.e., $f_{\tau}(\tau)\pi(m|\tau)m(t)$, is greater than the term in the second integral, i.e., $f_{\tau}(\tau)\pi(m|\tau)m(t)\frac{m(s)}{\mu}$. Therefore, as we increase $s$ and therefore we increase $\tilde{\tau}_m(s)$, we are increasing the range of the first integral and decreasing the range of the second, thus, increasing $\frac{\de}{\de s}\bar{c}_{t,\pi}(s)$. Hence, $\bar{c}_{t,\pi}(s)$ is convex with respect to $s$.
		\end{proof}
		
		According to Lemma \ref{convex}, it is necessary and sufficient for $t$ to be a local minimizer of $\bar{c}_{t,\pi}(s)$ to be its global minimizer. Therefore, we should have $\frac{\de}{\de s}\bar{c}_{t,\pi}(s)|_t=0$ and we have the result by setting \eqref{cprime} at $t$ to 0.
\optv{arxiv}{
	\subsection{Proof of Theorem \ref{fullinfo}}
		Let us assume that we have a positive queue over the interval $[t_1,t_2]$, i.e., $q_{\tau,m}(t)>0$ for $t \in (t_1,t_2)$ and $q_{\tau,m}(t_1)=0$ and $q_{\tau,m}(t_2)=0$. Note that we do not have any assumptions on the queue length at other times. 
		In order not to have any profitable deviations for agents arriving in $t\in [t_1,t_2]$,  we should have $c'_{\tau,m}(t)=0$ for $t \in (t_1,t_2)$ to avoid profitable deviations by changing the position inside the  queue. It implies the following.
		\begin{align*}
		c'_{\tau,m}(t)= \frac{m(t)}{\mu}-c=0 \Rightarrow m(t)=\mu c, \ t \in (t_1,t_2)
		\end{align*}
		% 			Therefore, we have the following equation for $m(t)$.
		% 			\begin{align*}
		% 				m(t)=
		% 				\mu c , \quad t \in (t_1,t_2)
		% 			\end{align*}
		Since the queue size is $0$ at $t_2$, an agent arriving at $t_2$ does not have any incentives for arriving later. Furthermore, in order for an agent arriving at $t_1$ not to have profit by arriving earlier, it is sufficient to have $t_1\leq \tau$.  This condition implies that we can not have multiple queues in the full information equilibrium. 
		
	  We can calculate the queue length at $t$ as follows.
		\begin{align*}
		q_{\tau,m}(t)=\int_{t_1}^{t}m(t)\de t-\mu(t-\tau)^+=c \mu(t-t_1)-\mu(t-\tau)^+
		\end{align*}
		%After $\tau$, the queue has rate $ q'_{\tau,m}(t)=-(1-c) \mu$. Therefore, starting from the value of $q_{\tau,m}(\tau)=c \mu(\tau-t_1)$ and decreasing with rate $-(1-c) \mu$, we should have $q_{\tau,m}(t_2)=0$.
		Setting the queue at $t_2$ to 0 will give us the equation below that  relates $t_1$ and $t_2$ to $\tau$.
		\begin{align*}
		\tau=ct_1+(1-c)t_2
		\end{align*}
		Since all agents must arrive between $[t_1,t_2]$, we have
		\begin{align*}
		c\mu(t_2-t_1)=1, 
		\end{align*}
		and therefore, we have
		\begin{align*}
		t_1=\tau-\frac{1-c}{c\mu},\quad
		t_2=\tau+\frac{1}{\mu}
		\end{align*}

	\subsection{Proof of Lemma \ref{mgrmu}}
		Assume we have a delta function of size $a$ at some time $t$ in the arrival process. We will show that the agent arriving at time $t$ has a profit by arriving slightly before $t$ at $s=t-\de t$. Note that we have $q_{\tau,m}(t-\de t)=q_{\tau,m}(t)-a$ for every $\tau$. The average cost of arriving at time $t$ is
		\begin{align*}
		\bar{c}(t)=\int_{\tau=0}^{\infty}(\frac{q_{\tau,m}(t)}{\mu}-c(t-\tau)+(t-\tau)^+)f_{\tau}(\tau)\de  \tau
		\end{align*}
		On the other hand, the average cost of arriving at time $s=t-\de t$ is 
		\begin{align*}
		\bar{c}(t-\de t)=&\int_{\tau=0}^{\infty}(\frac{q_{\tau,m}(t-\de t)}{\mu}-c(t-\de t-\tau) \nonumber \\&+(t-\de t-\tau)^+)f_{\tau}(\tau)\de \tau
		\end{align*}
		Subtracting the two will result in the following.
		\begin{align*}
		\bar{c}(t)-\bar{c}(t-\de t)=\frac{a}{\mu}-c \de t&+\de t \mathbf{1}(t \geq \tau)>0, \nonumber \\&  \text{for $\de t$ small enough.}
		\end{align*}
		Therefore, we can not have a delta function in the arrival process. 
		
		Next, assume $m(t)>\mu$ for some time $t$, which due to piecewise continuity of $m$, implies that $m(s)>\mu$ in a neighborhood of $t$. If $m$ is not continuous at $t$, we consider a point in the neighborhood of $t$ where $m$ is continuous.  Therefore,  
		we have $q(t)>0$ in a neighborhood of $t$. 
		%Therefore, the cost is differentiable for an interval of $s\leq t$.
		Let us denote $t_0$ to be the latest time before $t$ that $q(t_0)=0$. We have   $q(t)=\int_{l=t_0}^t m(l)\de l- \mu (t-\max(t_0,\tau))\mathbf{1}(t \geq \tau)$. Therefore, we can write the derivative of the average value of the cost as follows. 
		\begin{align*}
		\bar{c}^{\prime}(t)
		&= \frac{\de }{\de t}(\int_{\tau=0}^{\infty}(\int_{l=t_0}^t \frac{m(l)}{\mu}\de l-(t-\max(t_0,\tau))\mathbf{1}(t \geq \tau)\\ &\hspace{1.5cm}+(t-\tau)^+-c(t-\tau))f_{\tau}(\tau)\de \tau=\frac{m(t)}{\mu}-c
		\end{align*}
		Setting $\bar{c}^{\prime}(t)=0$ results in $m(t)=c \mu<\mu$. Therefore, we can not have $m(t)>\mu$.

	\subsection{Proof of Theorem \ref{noinfo}}
		%Notice that we can not have the server to work idle at any $s<t$ for $\tau=\tilde{\tau}_m(t)$. Otherwise, we can increase $\tau_o(t)$ and not affect the queue size at $t$. That is why the above equation always holds and we do not have the issue of queue size being zero and the server working idle in the middle of the arrival process. 
		We define $\bar{c}(t)$ to be the average value of the cost $c_{\tau,m}(t)$ with respect to $\tau$ using $f_{\tau}(\cdot)$.  
		In order to have an equilibrium, each agent arriving at time $t$ should to be acting rationally by doing so. Therefore, we should have $\bar{c}^{\prime}(t)=0$ for every $t$ that $m(t)> 0$ in a neighborhood of $t$. If $m(t)> 0$ in a right neighborhood of $t$, the right derivative of the expected cost should be zero and the left derivative should be non-positive. Similar rule applies for the left neighborhoods. 
		
		In Lemma \ref{mgrmu}, we proved that in order to satisfy incentive constraints in the no information case, we can never have $m(t)>\mu$, and $m(t)$ can never include a delta function. Therefore, we have  $m(t)\leq \mu$ for all $t$. 
		Also, according to Lemma \ref{taucont}, we know $\tilde{\tau}_m(t)$ is continuous and differentiable.  Therefore, the derivative of the average value of the cost, $\bar{c}^{\prime}(t)$ is given as follows. 
		\begin{align}
		\bar{c}&^{\prime}(t)
		=\frac{\de }{\de t}\int_{\tau}c_{\tau,m}(t)f_{\tau}(\tau)d \tau \nonumber 
		\\=& \frac{\de }{\de t}(\int_0^{\tilde{\tau}_m(t)}(1-c)(t-\tau)f_{\tau}(\tau)\de \tau\nonumber \\&+\int_{\tilde{\tau}_m(t)}^{\infty}  (\int_{l=-\infty}^t\frac{m(l)}{\mu}\de l-c(t-\tau))f_{\tau}(\tau)\de \tau)\nonumber
		\\
		=&\tilde{\tau}^{\prime}_m(t)(1-c)(t-\tilde{\tau}_m(t))f_{\tau}(\tilde{\tau}_m(t))\nonumber \\	&-\tilde{\tau}^{\prime}_m(t)(\int_{l=-\infty}^t\frac{m(l)}{\mu}\de l-c(t-\tilde{\tau}_m(t)))f_{\tau}(\tilde{\tau}_m(t)) \nonumber \\&+\int_0^{\tilde{\tau}_m(t)}(1-c)f_{\tau}(\tau)\de  \tau \nonumber 
		+\int_{\tilde{\tau}_m(t)}^{\infty}  (\frac{m(t)}{\mu}-c)f_{\tau}(\tau)\de \tau \nonumber 
		\\
		=&\int_0^{\tilde{\tau}_m(t)}(1-c)f_{\tau}(\tau)\de  \tau+\int_{\tilde{\tau}_m(t)}^{\infty}  (\frac{m(t)}{\mu}-c)f_{\tau}(\tau)\de \tau \nonumber
		\\
		=&1-e^{-\lambda \tilde{\tau}_m(t)}+\frac{m(t)}{\mu}e^{-\lambda \tilde{\tau}_m(t)}-c \nonumber
		\\
		&=1-c-e^{-\lambda (t-\int_{l=-\infty}^{t}\frac{m(l)}{\mu}\de  l)^+}(1-\frac{m(t)}{\mu}) \nonumber
		\end{align}
		%We can calculate $\bar{c}^{\prime}(t)$ as follows.
		%\ksmargin{Eq (15) seems incorrect; it does not take into account the fact that $\tau_0$ also depends on $t$ when taking derivative. Need to use Leibniz rule}
		%\begin{align}
		%  \bar{c}^{\prime}(t)
		%  &=\int_0^{\tilde{\tau}_m(t)}f_{\tau}(\tau)(1-c)d\tau +\int_{\tilde{\tau}_m(t)}^{\infty}  f_{\tau}(\tau)(\frac{m(t)}{\mu}-c)d\tau\\
		%  &=1-e^{-\lambda \tilde{\tau}_m(t)}+\frac{m(t)}{\mu}e^{-\lambda \tilde{\tau}_m(t)}-c\\
		%  &=1-c-e^{-\lambda (\int_0^t(1-\frac{m(s)}{\mu})ds)^+}(1-\frac{m(t)}{\mu}) \label{cprime-noinfo}
		%  \end{align}
		
		% Similarly, if $\tilde{\tau}_m(t)=0$ and we have $\int_0^t m(s)ds> \mu t$,  we have $\tau_0'(t,m)=0$ and the following holds. 
		% \begin{align}
		%      \bar{c}^{\prime}(t)
		%      =&\frac{d}{dt}\int_{\tau}c_{\tau,m}(t)f_{\tau}(\tau)d \tau \\=& \frac{d}{dt}(\int_{0}^{\infty}  (\int_{s=0}^t\frac{m(s)}{\mu}ds-c(t-\tau))f_{\tau}(\tau)d\tau)\\
		%      =&\int_{0}^{\infty}  (\frac{m(t)}{\mu}-c)f_{\tau}(\tau)d\tau\\
		%      =&\frac{m(t)}{\mu}-c=1-c-e^{-\lambda (\int_0^t(1-\frac{m(s)}{\mu})ds)^+}(1-\frac{m(t)}{\mu}) \label{cprime-noinfo}
		%      \end{align}
		%      Note that if $\tilde{\tau}_m(t)=0$ and we have $\int_0^t m(s)ds= \mu t$, $\tilde{\tau}_m(t)$ is not differentiable and we consider right and left derivatives instead. But one can see that the right derivative follows the first case above and the left derivative follows the second one, thus, satisfying the equation in both cases. 
		
		Setting $\bar{c}^{\prime}(t)=0$ will result in the following.
		\begin{align}
		e^{-\lambda (t-\int_{l=-\infty}^{t}\frac{m(l)}{\mu}\de l)^+}(1-\frac{m(t)}{\mu})=1-c
		\label{m-noinfo}
		\end{align}
		Equation \eqref{m-noinfo} holds for all $t$ such that we have $m(t)>0$. Note that if $m(s)=0$ for $s<t$ and $m(s)>0$ for $s\geq t$, the left derivative of $\bar{c}(t)$ is non-positive given equation \eqref{m-noinfo} holds for $t$. This implies that, as we increase $t$, we can have discontinuity in $m(t)$ from 0 to a non zero value. This is not the case for right neighborhoods with zero arrivals, i.e., $m(s)=0$ for an interval of $s>t$ and $m(s)>0$ for an interval of $s\leq t$ . In this case,  the right derivative will be  non-positive if \eqref{m-noinfo} holds for $t$. However, we need the right derivative to be positive for the agents to not have profitable deviations. Hence, whenever we have $m(s)=0$ for an interval of $s>t$, we should have $m(t)=0$, i.e., $m(t)$ must be continuous when transitioning to zero from non zero values. Also, note that the assumption of $m(t)\leq \mu$ clearly holds for any $m(t)$ satisfying equation \eqref{m-noinfo}. 
		%We further note that for $t\leq 0$, equation \eqref{m-noinfo} implies $m(t)=\mu c$. 
		Therefore, we have the following.
		% We further notice that if $\tilde{\tau}_m(t)=0$, we know that $\tilde{\tau}_m(s)=0$ for all $s\leq t$ and equation \eqref{m-difeq} implies that $m(s)=\mu c<\mu$ for all $s\leq t$. This is not consistent with $q_0(t,m)>0$. Therefore, we always have $\tilde{\tau}_m(t)>0$ and therefore, $\tilde{\tau}_m(t)=\int_0^t(1-\frac{m(s)}{\mu})ds$.  
		
		If we take the derivative of equation \eqref{m-noinfo} w.r.t. $t$ for $t\geq \tilde{t}$ ($\tilde{t}$ is defined in Lemma \ref{taucont}), we have
		\begin{align}
		&e^{-\lambda (t-\int_{l=-\infty}^{t}\frac{m(l)}{\mu}\de l)}(\lambda(1-\frac{m(t)}{\mu})^2\nonumber +\frac{m^{\prime}(t)}{\mu})=0\\
		&\Rightarrow\  -\lambda(1-\frac{m(t)}{\mu})^2-\frac{m^{\prime}(t)}{\mu}=0 \ \Rightarrow \ \frac{\de m}{(\mu-m)^2}=-\frac{\lambda}{\mu}\de t \nonumber\\
		&\Rightarrow\frac{1}{\mu-m}=\frac{-\lambda t +\beta}{\mu}   \Rightarrow m(t)=\mu- \frac{\mu}{\beta-\lambda t} \label{m-difeq}
		\end{align}
		In order to derive constant $\beta$, we assume that $m(t)$ is $0$ outside of an interval of $[t_1,t_2]$. If $\tilde{t}>0$ then we must have $t_1<0$. For now, we assume $\tilde{t_1}=0$ and therefore, $t_1\geq 0$. We must have $m(t_2)=0$ as mentioned in the discussions above. Also, since $\int_0^{t_2}m(t)\de t=1$, we have $\tilde{\tau}_m(t_2)=t_2-\frac{1}{\mu}$. Therefore, according to equation   \eqref{m-noinfo}, we have the following for $t_2$.
		\begin{align*}
		&e^{-\lambda (t_2-\frac{1}{\mu})}=1-c\\
		& \Rightarrow \ t_2=\frac{-\ln(1-c)}{\lambda}+\frac{1}{\mu}
		\end{align*}
		and we know $m(t_2)=0$ which will give us $\beta$ as follows.
		\begin{align*}
		&\mu- \frac{\mu}{\beta-\lambda t_2}=0 \Rightarrow\  \beta=\lambda t_2+1\\
		&\beta=-\ln(1-c)+\frac{\lambda}{\mu}+1.
		\end{align*}
		On the other hand, we must have $\int_{t_1}^{t_2}m(t)=1$, which results in the following equation to derive $t_1$.
		\begin{align}
		&\int_{t_1}^{t_2}m(t)\de t =\int_{t_1}^{t_2}(\mu- \frac{\mu}{\beta-\lambda t})\de t=1 \nonumber  \\
		&\Rightarrow \mu(t_2-t_1)-\frac{\mu}{\lambda}\ln(\lambda(t_2- t_1)+1 )=1 \nonumber\\
		&\Rightarrow \ln(1-c)+\lambda t_1+\ln(\frac{\lambda}{\mu}-\ln(1-c)-\lambda t_1+1 )=0 \label{t1pos}
		\end{align}
		If $t_1$ derived from the above equation is non-negative, then the equilibrium is characterized. Next, we consider the possibility of  $t_1\leq0$, which results in $\tilde{t}>0$. For $t\leq\tilde{t}$, $\tilde{\tau}_m(t)=0$ and according to \eqref{m-noinfo} we have $1-\frac{m(t)}{\mu}=1-c$ and 
		% 		\begin{align}
		% 			1-\frac{m(t)}{\mu}=1-c
		% 			\label{m-noinfo-gent}
		% 		\end{align}
		therefore, 
		we must have $m(t)=\mu c$ for  $t_1\leq t \leq \tilde{t}$. The queue size must be 0 at $\tilde{t}$ if $\tau=0$, because for $t>\tilde{t}$, we have $\tilde{\tau}_m(t)>0$. This results in the following.
		\begin{align*}
		& \mu c (\tilde{t}-t_1)=\mu\tilde{t} \Rightarrow \tilde{t}=-\frac{ c}{1- c} t_1
		\end{align*}
		On the other hand, since $\tilde{\tau}_m(t)>0$ for $t>\tilde{t}$ and $\tilde{\tau}_m(\tilde{t})=0$, $m(t)$ follows equation \eqref{m-difeq} for $t\geq \tilde{t}$ and we have $m(\tilde{t})=\mu-\frac{\mu}{\beta-\lambda \tilde{t}}$. Therefore, we have the following. 
		\begin{align}
		& m(\tilde{t})=\mu-\frac{\mu}{\beta-\lambda \tilde{t}}=\mu c \nonumber  \\
		& \Rightarrow 1-\frac{1}{\lambda(t_2+\frac{c}{1-c} t_1)+1}= c \nonumber  \\
		& \Rightarrow  \lambda((1-c)t_2+c t_1)=c\nonumber \\
		& \Rightarrow  -(1-c)\ln(1-c)+\frac{(1-c)\lambda}{\mu}+\lambda c t_1=c \nonumber 
		\\
		& \Rightarrow  t_1=\frac{1-c}{\lambda c}\ln(1-c)-\frac{1-c}{\mu c}+\frac{1}{\lambda} \label{t1neg}
		\end{align}
		If the value of $t_1$ above is negative, the no information equilibrium is characterized. Notice that we might have two types of no information equilibrium, one with negative $t_1$ and one with a positive one if the value of $t_1$ satisfying equations \eqref{t1pos} and \eqref{t1neg} is positive and negative, respectively.
	}

	\subsection{Proof of Theorem \ref{fulinfointerval}}
		Consider any $m$  in the support of $\pi(\cdot|\tau)$.  We show $m(t)$ as $m(t)=\mu c+\delta(t)$, where $\delta(t)$ is defined over $[\underline{t}_{\tau}, \bar{t}_{\tau}]$. Since we have $\int_{\underline{t}_{\tau}}^{\bar{t}_{\tau}}m(t)\de t=1$ and $ \bar{t}_{\tau}-\underline{t}_{\tau}\leq \frac{1}{\mu c}$, we must have $\int_{\underline{t}_{\tau}}^{\bar{t}_{\tau}}\delta(t)\de t\geq 0$. Using Lemma \ref{obedcon} we have the following.
		\begin{align}
		&(1-c)\int_{\tau,m}f_{\tau}(\tau)\pi(m|\tau)m(t)\mathbf{1}(\tau\leq \tilde{\tau}_m(t)) \de m+ \nonumber \\& \frac{1}{\mu}\int_{\tau,m} f_{\tau}(\tau)\pi(m|\tau)m(t)(m(t)-\mu c)\mathbf{1}(\tau>\tilde{\tau}_m(t)) \de m  =0 \nonumber
		% \\
		% &(1-c)\int_m\pi(m)m(t)p(\tau\leq \tilde{\tau}_m(t)|t,m) dm \nonumber \\& +\frac{1}{\mu}\int_m  \pi(m)(\delta(t)^2+\mu c \delta(t)) p(\tau>\tilde{\tau}_m(t)|t,m) dm=0\\
		%      &\int_{\tau,m}f_{\tau}(\tau)\pi(m|\tau)\int_t m(t)^2dt  dm d\tau =\mu c \int_{\tau,m}f_{\tau}(\tau)\pi(m|\tau) \int_t m(t)dt dm d\tau
		%      \\& \Rightarrow \int_{\tau,m}f_{\tau}(\tau)\pi(m|\tau)\int_t (\mu^2 c^2+2\mu c \delta(t)+\delta(t)^2)dt  dm d\tau=\mu c
		%      \\& \Rightarrow \mu c+\int_{\tau,m}f_{\tau}(\tau)\pi(m|\tau)\int_t \delta(t)^2dt  dm d\tau=\mu c
		%      \\& \Rightarrow \mathbb{E}[ \int_t \delta^2(t)]=0 \ \Rightarrow \delta(t)=0 \quad wp. \ 1
		\end{align}
		Since $\tilde{\tau}_m(t)$ is increasing in $t$, we can define its inverse by $\tilde{t}_m(\tau)$, i.e., we have $q_{\tau,m}(t)> 0$ for $\underline{t}_{\tau}\leq t< \tilde{t}_m(\tau)$ and $q_{\tau,m}(t)= 0$ for $t\geq  \tilde{t}_m(\tau)$. We have
		\begin{align*}
		&\frac{1}{\mu}\int_{\tau,m}f_{\tau}(\tau)\pi(m|\tau)\int_t m(t)(\mu(1-c)\mathbf{1}(t\geq \tilde{t}_m(\tau)) \nonumber \\&\quad +(m(t)-\mu c)\mathbf{1}(t< \tilde{t}_m(\tau)))\de \tau \de m \de t=0\\
		&\frac{1}{\mu}\int_{\tau,m}f_{\tau}(\tau)\pi(m|\tau)\int_t (\delta(t)+\mu c)(\mu(1-2c+c)\nonumber \\& \ \quad \quad \quad \mathbf{1}(t\geq \tilde{t}_m(\tau)) +\delta(t)\mathbf{1}(t< \tilde{t}_m(\tau)))\de \tau \de m \de t=0\\
		&\frac{1}{\mu}\int_{\tau,m}f_{\tau}(\tau)\pi(m|\tau)\int_{\underline{t}_{\tau}}^{\bar{t}_{\tau}} (\mu c\delta(t)+\mu^2c^2\mathbf{1}(t\geq \tilde{t}_m(\tau))\nonumber +\\&\ \mu(1-2c)m(t)\mathbf{1}(t\geq \tilde{t}_m(\tau))+\delta(t)^2\mathbf{1}(t< \tilde{t}_m(\tau)))\nonumber\\& 
		\hspace{6.7cm}	\de \tau \de m \de t
		=0
		\end{align*}
		We notice that all of the elements of the above integral are greater than or equal to zero. Therefore, they must all be zero for the sum to be zero. Hence, we have
		\begin{align*}
	%	&\int_{\underline{t}_{\tau}}^{\bar{t}_{\tau}} \mu c\delta(t) \de t=0\\
		&\int_{\underline{t}_{\tau}}^{\bar{t}_{\tau}} \mu^2c^2\mathbf{1}(t\geq \tilde{t}_m(\tau)) \de t=0\\
		& \int_{\underline{t}_{\tau}}^{\bar{t}_{\tau}} \mu(1-2c)m(t)\mathbf{1}(t\geq \tilde{t}_m(\tau)) \de t=0\\
		&\int_{\underline{t}_{\tau}}^{\bar{t}_{\tau}} \delta(t)^2\mathbf{1}(t< \tilde{t}_m(\tau))) \de t=0
		\end{align*}
		Therefore, we must have $\delta(t)=0$ for all $t\in [\underline{t}_{\tau},\bar{t}_{\tau}], m$. Hence, $m(t)=\mu c$ and thus,  $\bar{t}_{\tau}-\underline{t}_{\tau}= \frac{1}{\mu c}$, i.e., the time span of the arrival processes are equal to the one in the full information equilibrium.
		We must also have $\mathbf{1}(t\geq \tilde{t}_m(\tau))=0$ for all $t\in [\underline{t}_{\tau},\bar{t}_{\tau}],m,\tau$, which is consistent with assumption (c). 
		Therefore, we must have $\underline{t}_{\tau}=\tau-\frac{1-c}{c \mu}$ and $\bar{t}_{\tau}=\tau+\frac{1}{\mu}$. Hence, $\pi(\cdot|\tau)$ is supported only over the full information equilibrium arrival process and the theorem is proved.

	\optv{arxiv}{
\subsection{Proof of Lemma \ref{obedGPM}}
		If the planner restricts his attention to the set of signaling strategies that satisfy assumptions (b) and (c), we have $q_{\tau,m}(t)=\int_{l=-\infty}^t m(l)\de l-\mu(\tau-t)^+$. Therefore, we have the following for $\bar{c}_{t,\pi}(s)$ and its derivative. 
		\begin{align*}
		\bar{c}&_{t,\pi}(s)= \frac{1}{\mu \bar{m}(t)}\int_m\int_{\tau=\underline{\tau}(t)}^{\bar{\tau}(t)}(\int_{l=-\infty}^s m(l)\de l-\mu(s-\tau)^+ \\&\quad \ +\mu c(\tau-s)+\mu (s-\tau)^+) f_{\tau}(\tau)\pi(m|\tau)m(t) \de \tau    \de m c  \nonumber 
		\\ 
		=&\frac{1}{\mu \bar{m}(t)}\int_m\int_{\tau=\underline{\tau}(t)}^{\bar{\tau}(t)}f_{\tau}(\tau)\pi(m|\tau)m(t)\nonumber
		\\& \hspace{2.2cm}(\int_{l=-\infty}^s m(l)\de l -\mu c s) \ \de \tau    \de m +c\mathbb{E}(\tau|t) \nonumber
		\end{align*}
		\begin{align*}
		\frac{\de}{\de s}\bar{c}&_{t,\pi}(s)
		=\frac{1}{\mu \bar{m}(t)}\int_m\int_{\tau=\underline{\tau}(t)}^{\bar{\tau}(t)}f_{\tau}(\tau)\pi(m|\tau)m(t)\nonumber
		\\& \hspace{3.6cm}( m(s)-\mu c) \ \de \tau    \de m  \nonumber
		\end{align*}
		According to Lemma \ref{obedcon}, if we set $\frac{\de}{\de s}\bar{c}_{t,\pi}(s)|_t=0$ we get the result. 
		% 	 we can  derive the derivative of the cost $c'_{\tau}(t,q)$ as follows.
		% 	\begin{align*}
		% 	c'_{\tau,m}(t)&=\left\{ \begin{array}{cc}
		% 	\frac{m(t)}{\mu}-c,  & \forall t\leq  \bar{t}_{\tau}\\ %\max(\bar{t}_{\tau},t^0_{\tau}(m)) \\
		% 	1-c,  & \forall t> \bar{t}_{\tau}
		% 	%\max(\bar{t}_{\tau},t^0_{\tau}(m)) 
		% 	\end{array}\right.
		% 	\end{align*}

		% 	Therefore, the average value of the cost $	\bar{c}'_{t}(t)$
		% 	\begin{align*}
		% 	\frac{d}{ds}\bar{c}_{t,\pi}(s)|_t=0=& \frac{\int_{\underline{\tau}(t)}^{\bar{\tau}(t)}\int_{m}f_{\tau}(\tau)\pi(m|\tau)m(t)\frac{m(t)}{\mu}d\tau dm}{\int_{\underline{\tau}(t)}^{\bar{\tau}(t)}\int_{m}f_{\tau}(\tau)\pi(m|\tau)m(t)d\tau dm}-c,
		% 	\end{align*}

	\subsection{Proof of Lemma \ref{objGPM}}
		\begin{align*}
		\bar{s}(&\pi)=
		\int_{t}\int_{\tau,m}f_{\tau}(\tau)\pi(m|\tau)m(t)c_{\tau,m}(t) \de \tau \de m \de t\\
		&=\int_{t}\int_{\tau,m}f_{\tau}(\tau)\pi(m|\tau)m(t) \nonumber \\& \hspace{0.7cm}(\frac{q(t)}{\mu}+c(\tau-t)^++(1-c)(t-\tau)^+) \de \tau \de m \de t\\
		&=\int_{t}\int_{\tau,m}f_{\tau}(\tau)\pi(m|\tau)m(t)\nonumber (\frac{\int_{l=-\infty}^tm(l)dl-\mu(t-\tau)^+}{\mu}\\& \hspace{3cm}+c(\tau-t)+(t-\tau)^+)\de \tau \de m \de t\\
		&=\int_{t}\int_{\tau,m}f_{\tau}(\tau)\pi(m|\tau)m(t)(\frac{\int_{l=-\infty}^tm(l)\de l}{\mu}+c(\tau-t))\\& \hspace{7cm} \de \tau \de m \de t\\
		&=\frac{1}{\mu}\int_{\tau}f_{\tau}(\tau)(\int_{t=\underline{t}_{\tau}}^{\bar{t}_{\tau}}\int_{s=\underline{t}_{\tau}}^t(R_{m,\tau}(t,s)-\mu c\bar{m}_{\tau}(t))\de s \de t\nonumber \\&\hspace{3.2cm}+\mu c (\tau-\underline{t}_{\tau}))\de\tau
		\end{align*}
	
\subsection{Proof of Theorem \ref{fulinfores}}
		Suppose $m(t)$ is in the support of $\pi(\cdot|\tau)$.  Assume $\bar{t}_{\tau}-\underline{t}_{\tau}=T$. We show $m(t)$ as $m(t)=\mu c+\delta(t)$. Since we have $\int_{\underline{t}_{\tau}}^{\bar{t}_{\tau}}m(t)\de t=1$, we must have $\int_{\underline{t}_{\tau}}^{\bar{t}_{\tau}}\delta(t)\de t=1-\mu c T\geq 0$. Lemma \ref{obedGPM} results in the following.
		\begin{align*}
	& \int_{\tau,m}\hspace{-0.1cm}f_{\tau}(\tau)\pi(m|\tau)\hspace{-0.1cm}\int_t (\mu^2 c^2+2\mu c \delta(t)+\delta(t)^2)\de t  \de m \de \tau\nonumber=\mu c
		\\& \Rightarrow \mu c (1-
		\mu c T)+\int_{\tau,m}f_{\tau}(\tau)\pi(m|\tau)\int_t \delta(t)^2\de t  \de m \de \tau=0
		\\& \Rightarrow \mathbb{E}[ \int_t \delta^2(t)\de t]=0 \ \Rightarrow \delta(t) =0 \quad wp. \ 1
		\\& \qquad  \mu c T=1 \Rightarrow  T=\frac{1}{\mu c}
		\end{align*}
		Therefore, we have $m(t)=\mu c$ with probability one and $\bar{t}_{\tau}-\underline{t}_{\tau}=\frac{1}{\mu c}$. Therefore, we must have $m(t)$ to be the full information equilibrium, i.e., $\underline{t}_{\tau}=\tau-\frac{1-c}{c \mu}$ and $\bar{t}_{\tau}=\tau+\frac{1}{\mu}$. Therefore, $\pi(\cdot|\tau)$ is supported only over the full information equilibrium arrival process and the  result is proved. 
}	
	
	%	\bibliographystyle{ieeetr}
	%	\bibliography{../../../../bib/ksmain,../../../../bib/savla}

	\bibliographystyle{IEEEtran}
	%\bibliography{achilleas18abrv,achilleas18_own,achilleas18_control,Nasimeh}
	\optv{arxiv}{% Generated by IEEEtran.bst, version: 1.14 (2015/08/26)

}
		\optv{submission}{\input{root_bib.bbl}}

\end{document}